\documentclass[10pt, final, journal]{IEEEtran} 
\usepackage{listings}
\usepackage{amsmath}
\usepackage{amsthm}
\usepackage{tikz}
\usepackage{caption}
\usepackage{array}
\usepackage{mdwmath}
\usepackage{multirow}
\usepackage{mdwtab}
\usepackage{eqparbox}
\usepackage{amsfonts}
\usepackage{tikz}
\usepackage{multirow,bigstrut,threeparttable}
\usepackage{amsthm}
\usepackage{array}
\usepackage{bbm}
\usepackage{subfigure}
\usepackage{epstopdf}
\usepackage{mdwmath}
\usepackage{mdwtab}
\usepackage{eqparbox}
\usepackage{tikz}
\usepackage{latexsym}
\usepackage{cite}
\usepackage{amssymb}
\usepackage{bm}
\usepackage{amssymb}
\usepackage{graphicx}
\usepackage{mathrsfs}
\usepackage{epsfig}
\usepackage{psfrag}
\usepackage{setspace}
\usepackage{hyperref}
\usepackage{algorithm}
\usepackage{algpseudocode}
\usepackage{stfloats}

\newtheorem{remark}{Remark}

\newtheorem{theorem}{Theorem}

\newtheorem{lemma}{Lemma} 
\newtheorem{corollary}{Corollary}

\newtheorem{definition}{Definition}

\def \cZ {\mathcal{Z}}
\def \bP {\mathbb{P}}
\def \bE {\mathbb{E}}

\def \cM {\mathcal{M}}

\begin{document}

\title{Maximum Likelihood Estimation of Functionals of Discrete Distributions}

\author{Jiantao~Jiao,~\IEEEmembership{Student Member,~IEEE},~Kartik~Venkat,~\IEEEmembership{Student Member,~IEEE},~Yanjun~Han,~\IEEEmembership{Student Member,~IEEE}, and Tsachy~Weissman,~\IEEEmembership{Fellow,~IEEE}

\thanks{Manuscript received Month 00, 0000; revised Month 00, 0000; accepted Month 00, 0000. Date of current version Month 00, 0000. This work was supported in part by the Center for Science of Information (CSoI) under grant agreement CCF-0939370. Materials of this paper were presented in part at the 2015 International Symposium on Information Theory, Hong Kong, China. Personal use of this material is permitted. However, permission to use this material for any other purposes must be obtained from the IEEE by sending a request to pubs-permissions@ieee.org. }
\thanks{Jiantao Jiao, Yanjun Han, and Tsachy Weissman are with the Department of Electrical Engineering, Stanford University, CA, USA. Email: \{jiantao,yjhan, tsachy\}@stanford.edu.} 
\thanks{Kartik Venkat was with the Department of Electrical Engineering, Stanford University. He is currently with PDT Partners. Email: kvenkat@alumni.stanford.edu. }
}

\date{\today}

\vspace{-10pt}

\maketitle

\begin{abstract}
We consider the problem of estimating functionals of discrete distributions, and focus on tight (up to universal multiplicative constants for each specific functional) nonasymptotic analysis of the worst case squared error risk of widely used estimators. We apply concentration inequalities to analyze the random fluctuation of these estimators around their expectations, and the theory of approximation using positive linear operators to analyze the deviation of their expectations from the true functional, namely their \emph{bias}.

We explicitly characterize the worst case squared error risk incurred by the Maximum Likelihood Estimator (MLE) in estimating the Shannon entropy $H(P) = \sum_{i = 1}^S -p_i \ln p_i$, and the power sum $F_\alpha(P) = \sum_{i = 1}^S p_i^\alpha,\alpha>0$, up to universal multiplicative constants for each fixed functional, for any alphabet size $S\leq \infty$ and sample size $n$ for which the risk may vanish. As a corollary, for Shannon entropy estimation, we show that it is necessary and sufficient to have $n  \gg S$ observations for the MLE to be consistent. In addition, we establish that it is necessary and sufficient to consider $n \gg S^{1/\alpha}$ samples for the MLE to consistently estimate $F_\alpha(P), 0<\alpha<1$. The minimax rate-optimal estimators for both problems require $S/\ln S$ and $S^{1/\alpha}/\ln S$ samples, which implies that the MLE has a strictly sub-optimal sample complexity. When $1<\alpha<3/2$, we show that the worst-case squared error rate of convergence for the MLE is $n^{-2(\alpha-1)}$ for infinite alphabet size, while the minimax squared error rate is $(n\ln n)^{-2(\alpha-1)}$. When $\alpha\geq 3/2$, the MLE achieves the minimax optimal rate $n^{-1}$ regardless of the alphabet size.

As an application of the general theory, we analyze the Dirichlet prior smoothing techniques for Shannon entropy estimation. In this context, one approach is to plug-in the Dirichlet prior smoothed distribution into the entropy functional, while the other one is to calculate the Bayes estimator for entropy under the Dirichlet prior for squared error, which is the conditional expectation. We show that in general such estimators do \emph{not} improve over the maximum likelihood estimator. No matter how we tune the parameters in the Dirichlet prior, this approach cannot achieve the minimax rates in entropy estimation. The performance of the minimax rate-optimal estimator with $n$ samples is essentially \emph{at least} as good as that of Dirichlet smoothed entropy estimators with $n\ln n$ samples. 
\end{abstract}


\section{Introduction}

Entropy and related information measures arise in information theory, statistics, machine learning, biology, neuroscience, image processing, linguistics, secrecy, ecology, physics, and finance, among other fields. Numerous inferential tasks rely on data driven procedures to estimate these quantities~(see, e.g. \cite{Olsen--Meyer--Bontempi2009impact,Pluim--Maintz--Viergever2003mutual,Viola--Wells1997alignment,Batina--Gierlichs--Prouff--Rivain--Standaert--Veyrat2011mutual,Hill1973diversity,Franchini--Its--Korepin2008Renyi}). We focus on two concrete and well-motivated examples of information measures, namely the Shannon entropy~\cite{Shannon1948}
\begin{equation}\label{eqn.entropy}
H(P) \triangleq \sum_{i = 1}^S - p_i \ln p_i,
\end{equation}
and the power sum $F_\alpha(P), \alpha > 0$:
\begin{equation}
F_\alpha (P) \triangleq \sum_{i = 1}^S p_i^\alpha, \alpha>0.
\end{equation}
The power sum $F_\alpha(P)$ functional often emerges in various operational problems~\cite{Breiman--Friedman--Stone--Olshen1984classification}. It also has connections to the R\'enyi entropy \cite{Renyi1961measures} $H_\alpha(P)$ via the formula $H_\alpha(P) = \frac{\ln F_\alpha(P)}{1-\alpha}$. 

Consider estimating the Shannon entropy $H(P)$ based on $n$ i.i.d. samples following unknown discrete distribution $P$ with \emph{unknown} alphabet size $S$. This problem has a rich history with extensive study in various fields ranging from information theory, statistics, neuroscience, physics, psychology, medicine, etc. We refer the reader to~\cite{Jiao--Venkat--Han--Weissman2015minimax} for a review. One of the most widely used estimators for this purpose is the Maximum Likelihood Estimator (MLE), which is simply the empirical entropy. The empirical entropy is an instantiation of the plug-in principle in functional estimation, where a point estimate of the parameter (distribution $P$ in this case) is used to construct an estimator for a {\em functional} of the parameter via the plug-in approach. The idea of using the MLE for estimating information measures of interest (in this case entropy), is not only intuitive, but has sound justification: {\em asymptotic efficiency}.

The beautiful theory of H\'ajek and Le Cam \cite{Hajek1970characterization,Hajek1972local,LeCam1986asymptotic} shows that, as the number of observed samples grows without bound while the finite parameter dimension (e.g., alphabet size) remains fixed, the MLE performs optimally in estimating any differentiable functional when the statistical model complies with the benign LAN (Local Asymptotic Normality) condition~\cite{LeCam1986asymptotic}. Thus, for finite dimensional problems, the problems of parameter and functional estimation are well understood in an asymptotic sense, and the MLE appears to be not only natural but also theoretically justified. But does it make sense to employ the MLE to estimate the entropy in most practical applications?

As it turns out, while asymptotically optimal in entropy estimation, the MLE is by no means sacrosanct in many real applications, especially in regimes where the alphabet size is comparable to, or even larger than the number of observations. It was shown that the MLE for entropy is strictly sub-optimal in the large alphabet regime~\cite{Paninski2003,Paninski2004}. Therefore, classical asymptotic theory does not satisfactorily address high dimensional settings, which are becoming increasingly important in the modern era of high dimensional statistics.

There has been a wave of recent research activities focusing on analyzing existing approaches of functional estimation, as well as proposing new estimators that are provably near optimal in the large alphabet regime. Paninski~\cite{Paninski2003} showed that the MLE needs $n \gg S$ samples to consistently estimate the Shannon entropy, and Paninski~\cite{Paninski2004} established the existence of a (non-explicit) estimator that only required $n \ll S$ samples. It implies that the MLE is strictly sub-optimal in terms of sample complexity. It was Valiant and Valiant~\cite{Valiant--Valiant2011} who first explicitly constructed a linear programming based estimator (later modified in~\cite{Valiant--Valiant2013estimating}) that achieves consistency in entropy estimation with $n \gg S/\ln S$ samples, which they also proved to be necessary. Valiant and Valiant~\cite{Valiant--Valiant2011power} constructed another approximation based estimator that achieved better theoretical properties than the linear programming ones, which was not yet shown to be minimax rate-optimal for all ranges of $S$ and $n$. The authors~\cite{Jiao--Venkat--Han--Weissman2015minimax} constructed the first minimax rate-optimal estimators for $H(P)$ and $F_\alpha(P),\alpha>0$ based on best polynomial approximation, which are agnostic to the alphabet size $S$. Utilizing the released MATLAB and Python packages of the estimators in~\cite{Jiao--Venkat--Han--Weissman2015minimax}, \cite{Jiao--Venkat--Han--Weissman2014beyond,jiao2016beyond} demonstrated that these minimax rate-optimal estimators can lead to significant performance boosts in various machine learning tasks. Wu and Yang~\cite{Wu--Yang2014minimax} independently applied the best polynomial approximation idea to entropy estimation and obtained the minimax rates. However, their estimator requires the knowledge of the alphabet size $S$. The approximation ideas proved to be very fruitful in Acharya \emph{et al.}~\cite{Acharya--Orlitsky--Suresh--Tyagi2014complexity}, Wu and Yang~\cite{wu2015chebyshev}, Han, Jiao, and Weissman~\cite{Han--Jiao--Weissman2016minimaxdivergence}, Jiao, Han, and Weissman~\cite{jiao2016minimaxl1distance}, Bu \emph{et al.}~\cite{bu2016estimation}, Orlitsky, Suresh, and Wu~\cite{orlitsky2016optimal}, Wu and Yang~\cite{wu2016sample}.  

The main contribution of this paper is an explicit characterization of the worst case squared error risk of estimating $H(P)$ and $F_\alpha(P)$ using the MLE up to a universal multiplicative constant for each specific functional, for all ranges of $S$ and $n$ in which the risk may vanish. Understanding the benefits and limitations of the MLE in a nonasymptotic setting serves two key purposes. First, the approach is a natural benchmark for comparing other more nuanced procedures for estimation of functionals. Second, performance analysis for the MLE reveals regimes where the problem is difficult, and motivates the development of improvements, which have been validated in~\cite{Paninski2003,Paninski2004,Valiant--Valiant2011,Valiant--Valiant2013estimating,Valiant--Valiant2011power,Jiao--Venkat--Han--Weissman2015minimax,Wu--Yang2014minimax,Acharya--Orlitsky--Suresh--Tyagi2014complexity}. As a byproduct of the analysis, we explicitly point out an equivalence between bias analysis of functional estimators using plug-in rules and approximation theory using positive linear operators. We believe these powerful tools introduced from approximation theory may have far reaching impacts in various applications in the information theory community. 

We mention that there exist numerous other approaches proposed in various disciplines to estimate entropy, many among which are difficult to analyze theoretically. Among them we mention the Miller--Madow bias-corrected estimator and its variants \cite{Miller1955,Carlton1969bias,Grassberger1988finite}, the jackknife estimator \cite{Zahl1977jackknifing}, the shrinkage estimator \cite{Hausser--Strimmer2009entropy}, the coverage adjusted estimator \cite{Chao--Shen2003nonparametric}, the Best Upper Bound (BUB) estimator \cite{Paninski2003}, the B-Splines estimator \cite{Daub--Steuer--Selbig--Kloska2004}, and~\cite{Grassberger2008entropy,Vinck--Battaglia--Balakirsky--Vinck--Pennartz2012estimation} etc. For a Bayesian statistician, a natural approach is to first impose a prior on the unknown discrete distribution before considering estimating entropy. The Dirichlet prior, being the conjugate prior to the multinomial distribution, appears to be particularly popular in the Bayesian approach to entropy estimation. Dirichlet smoothing may have two connotations in the context of entropy estimation: 
\begin{itemize}
\item \cite{Schurmann--Grassberger1996entropy,Schober2013some} One first obtains a Bayes estimate for the discrete distribution $P$, which we denote by $\hat{P}_B$, and then plugs it in the entropy functional to obtain the entropy estimate $H(\hat{P}_B)$.
\item \cite{Wolpert--Wolf1995}\cite{Holste--Grosse--Herzel1998bayes} One calculates the Bayes estimate for entropy $H(P)$ under Dirichlet prior for squared error. The estimator is the conditional expectation $\bE[H(P)|\mathbf{X}]$, where $\mathbf{X}$ represents the samples. 
\end{itemize}

Nemenman, Shafee, and Bialek~\cite{Nemenman--Shafee--Bialek2002entropy} argued in an intuitive way why Dirichlet prior is bad for entropy estimation and proposed to use mixtures of Dirichlet priors. Archer, Park, and Pillow~\cite{Archer--Park--Pillow2012bayesian} have come up with priors that perform better than the Dirichlet prior. Also see~\cite{Nemenman--Bialek--vanSteveninck2004entropy,Nemenman2011coincidences}. 

Another contribution of this paper is an explicit characterization of the worst case squared error risk of estimating $H(P)$ using the Dirichlet prior plug-in approach up to a universal multiplicative constant, for all ranges of $S$ and $n$ in which the risk may vanish. We show rigorously that neither of the two approaches utilizing the Dirichlet prior result in improvements over the MLE in the large alphabet regime. Specifically, these approaches require at least $n\gg S$ to be consistent, while the minimax rate-optimal estimators such as the ones in~\cite{Jiao--Venkat--Han--Weissman2015minimax}\cite{Wu--Yang2014minimax} only need $n \gg \frac{S}{\ln S}$ to achieve consistency.

The rest of the paper is organized as follows. We present the main results in Section~\ref{sec.mainresults}, discuss the fundamental ideas behind the proofs in Section~\ref{sec.ideas}, and detail the proofs in Section~\ref{sec.upperboundproof} and~\ref{sec.lowerboundproof}. Proofs of auxiliary lemmas are deferred to the appendices. 

\section{Preliminaries}\label{sec.preliminaries}

The Dirichlet distribution with order $S\geq 2$ with parameters $\alpha_1,\ldots,\alpha_S>0$ has a probability density function with respect to Lebesgue measure on the Euclidean space $\mathbb{R}^{S-1}$ given by
\begin{equation}
f \left(x_1,\cdots, x_{S}; \alpha_1,\cdots, \alpha_S \right) = \frac{1}{\mathrm{B}(\boldsymbol\alpha)} \prod_{i=1}^S x_i^{\alpha_i - 1}
\end{equation}
on the open $S-1$-dimensional simplex defined by:
\begin{align} 
&x_1, \cdots, x_{S-1} > 0 \\
&x_1 + \cdots + x_{S-1} < 1 \\ 
&x_S = 1 - x_1 - \cdots - x_{S-1} 
\end{align}
and zero elsewhere. The normalizing constant is the multinomial Beta function, which can be expressed in terms of the Gamma function:
\begin{equation}
\mathrm{B}(\boldsymbol\alpha) = \frac{\prod_{i=1}^S \Gamma(\alpha_i)}{\Gamma\left(\sum_{i=1}^S \alpha_i\right)},\qquad\boldsymbol{\alpha}=(\alpha_1,\cdots,\alpha_S).
\end{equation}

Assuming the unknown discrete distribution $P$ follows prior distribution $P \sim \text{Dir}(\boldsymbol\alpha)$, and we observe a vector $\mathbf{X} = (X_1,X_2,\ldots,X_S)$ with multinomial distribution $\mathsf{multi}(n;p_1,p_2,\ldots,p_S)$, then one can show that the posterior distribution $P_{P|\mathbf{X}}$ is also a Dirichlet distribution with parameters
\begin{equation}
{\boldsymbol\alpha} + \mathbf{X} = \left( \alpha_1 + X_1,\alpha_2 + X_2,\ldots,\alpha_S + X_S \right).
\end{equation}

Furthermore, the posterior mean (conditional expectation) of $p_i$ given $\mathbf{X}$ is given by~\cite[Example 5.4.4]{Lehmann--Casella1998theory}
\begin{equation}
\delta_i(\mathbf{X}) \triangleq \bE[p_i|\mathbf{X}] = \frac{\alpha_i + X_i}{n + \sum_{i = 1}^S \alpha_i}.
\end{equation}

The estimator $\delta_i(\mathbf{X})$ is widely used in practice for various choices of $\mathbf{\alpha}$. For example, if $\alpha_i = \frac{\sqrt{n}}{S}$, then the corresponding $\left( \delta_1(\mathbf{X}),\delta_2(\mathbf{X}),\ldots,\delta_S(\mathbf{X})\right)$ is the minimax estimator for $P$ under squared loss~\cite[Example 5.4.5]{Lehmann--Casella1998theory}. However, it is no longer minimax under other loss functions such as $\ell_1$ loss, which was investigated in~\cite{Han--Jiao--Weissman2015minimaxdistribution}. 

Note that the estimator $\delta_i(\mathbf{X})$ subsumes the MLE $\hat{p}_i = \frac{X_i}{n}$ as a special case, since we can take the limit $\boldsymbol\alpha\to 0$ for $\delta_i(\mathbf{X})$ to obtain MLE. We denote the empirical distribution by $P_n = (\hat{p}_1,\hat{p}_2,\ldots, \hat{p}_S)$. The Dirichlet prior smoothed distribution estimate is denoted as $\hat{P}_B$, where
\begin{equation}
\hat{P}_B = \frac{n}{n + \sum_{i = 1}^S \alpha_i} P_n + \frac{\sum_{i = 1}^S \alpha_i}{n+\sum_{i = 1}^S \alpha_i}  \frac{\boldsymbol\alpha}{\sum_{i = 1}^S \alpha_i}.    
\end{equation}
Note that the \emph{smoothed} distribution $\hat{P}_B$ can be viewed as a convex combination of the empirical distribution $P_n$ and the \emph{prior} distribution $\frac{\boldsymbol\alpha}{\sum_{i = 1}^S \alpha_i}$. We call the estimator $H(\hat{P}_B)$ the \emph{Dirichlet prior smoothed plug-in estimator}. 

Another way to apply Dirichlet prior in entropy estimation is to compute the Bayes estimator for $H(P)$ under squared error, given that $P$ follows Dirichlet prior. It is well known that the Bayes estimator under squared error is the conditional expectation. It was shown in Wolpert and Wolf~\cite{Wolpert--Wolf1995} that
\begin{align}
\hat{H}^{\mathsf{Bayes}} & \triangleq \bE[H(P)|\mathbf{X}] \nonumber \\
& = \psi\left( 1 + \sum_{i = 1}^S (\alpha_i + X_i)\right) \nonumber \\
& \quad  - \sum_{i = 1}^S \left( \frac{\alpha_i + X_i}{\sum_{i = 1}^S (\alpha_i + X_i)} \right) \psi(\alpha_i + X_i + 1),\label{eqn.bayesestimator}
\end{align}
where $\psi(z) \triangleq \frac{\Gamma'(z)}{\Gamma(z)}$ is the digamma function. We call the estimator $\hat{H}^{\mathsf{Bayes}}$ the \emph{Bayes estimator under Dirichlet prior}.

Throughout this paper, we observe $n$ i.i.d. samples from an unknown discrete distribution $P = (p_1, p_2, \ldots, p_S)$. We denote the $n$ samples as $n$ i.i.d. random variables $\{Z_i\}_{1\leq i\leq n}$ taking values in $\mathcal{Z} = \{1,2,\ldots,S\}$ with probability $(p_1,p_2,\ldots,p_S)$. Defining
\begin{equation}
X_i \triangleq \sum_{j = 1}^n \mathbbm{1}(Z_j = i),\quad 1\leq i \leq S,
\end{equation}
we know that $(X_1,X_2,\ldots,X_S)$ follows a multinomial distribution with parameter $(n;p_1,p_2,\ldots,p_S)$. Denote $h_j \triangleq \sum_{i = 1}^S \mathbbm{1}(X_i = j),\quad 0\leq j \leq n$. The Maximum Likelihood Estimator (MLE) for $H(P)$ and $F_\alpha(P)$ are defined, respectively, as $H(P_n)$ and $F_\alpha(P_n)$, with $P_n$ being the empirical distribution. We assume the functional $F(P)$ takes the form
\begin{align}\label{eqn.generalf}
F(P) = \sum_{i = 1}^S f(p_i).
\end{align} 
Then it is evident that the MLE $F(P_n)$ for estimating functional $F(P)$ in (\ref{eqn.generalf}) can be alternatively represented as the following linear function of $(h_0,h_1,\ldots,h_n)$:
\begin{equation}
F(P_n) = \sum_{j = 0}^n f\left( \frac{j}{n} \right) h_j.
\end{equation}

Recall that the risk function under squared error for any estimator $\hat{F}$ in estimating functional $F(P)$ may be decomposed as
\begin{align}
\mathbb{E}_{P} (F(P) - \hat{F})^2 & = (\mathbb{E}_P \hat{F} - F(P))^2 + \mathbb{E}_P \left( \hat{F} - \mathbb{E}_P \hat{F}\right)^2,
\end{align}
where $(\mathbb{E}_P \hat{F} - F(P))^2$ represents the squared bias, and $\mathbb{E}_P \left( \hat{F} - \mathbb{E}_P \hat{F}\right)^2$ represents the variance. The subscript $P$ means that the expectation is taken with respect to the distribution $P$ that generates the i.i.d.\ observations. We omit the subscript for the expectation operator $\mathbb{E}$ if the meaning of the expectation is clear from the context.

{\em Notation:} $a \wedge b$ denotes $\min\{a,b\}$, $a \vee b$ denotes $\max\{a,b\}$. For two non-negative series $\{a_n\},\{b_n\}$, notation $a_n \lesssim b_n$ means that there exists a positive universal constant $C<\infty$ such that $\frac{a_n}{b_n}\leq C$, for all $n$. The notation $a_n \asymp b_n$ is equivalent to $a_n \lesssim b_n$ and $b_n \lesssim a_n$. Notation $a_n \gg b_n$ means that $\liminf_{n\to \infty} \frac{a_n}{b_n} =\infty$. Throughout this paper, the notations $\lesssim, \gtrsim, \ll, \gg$ involve absolute constants that may only depend on $\alpha$ but not $S$ or $n$. We denote by $\mathcal{M}_S$ the space of discrete distributions with alphabet size $S$.

\section{Main results}\label{sec.mainresults}

\subsection{Estimating $F_\alpha(P)$}

We split the upper bounds and the lower bounds into two theorems, and present their succinct summaries in Corollary~\ref{cor.largerthanone} and~\ref{cor.l2ratesfalphalessone}.  

\begin{theorem}[Upper bounds]\label{thm.main}
We have the following upper bounds on the worst case squared error risk of MLE in estimating $F_\alpha(P)$:
\begin{enumerate}
\item $\alpha\geq 2$:
\begin{align}
\sup_{P \in \cM_S} \bE_P \left( F_\alpha(P_n) - F_\alpha(P) \right)^2  \leq \left( \frac{\alpha(\alpha-1)}{2n}\right)^2 + \frac{\alpha^2}{4n}.
\end{align}
\item $1<\alpha<2$:
\begin{align}
& \sup_{P \in \cM_S} \bE_P \left( F_\alpha(P_n) - F_\alpha(P) \right)^2 \nonumber \\
& \quad \leq \left ( \frac{4}{n^{\alpha-1}} \wedge  \frac{3S^{1-\alpha/2}}{n^{\alpha/2}} \wedge C_{\alpha,n} \frac{5S}{2n} \right)^2  + \frac{\alpha^2}{4n},
\end{align}
where $C_{\alpha,n} \triangleq n \omega_\varphi^2(x^\alpha,n^{-1/2})>0$ satisfies $\limsup_{n\to \infty}C_{\alpha,n}<\infty$ for $1<\alpha<2$, and $\omega_\varphi^2$ is the second-order Ditzian--Totik modulus of smoothness introduced in Section~\ref{sec.biasidea}. 
\item $1/2 \leq \alpha<1$:
\begin{align}
& \sup_{P \in \cM_S} \bE_P \left( F_\alpha(P_n) - F_\alpha(P) \right)^2 \nonumber \\
&\quad \leq \left(\frac{3S^{1-\alpha/2}}{2n^{\alpha/2}} \wedge  \frac{5S}{2n^{\alpha}} \right)^2 \nonumber \\
& \qquad + \left( \frac{10S^{2-2\alpha}}{n}+ \frac{120}{\alpha^2}\left(\frac{S}{n^{2\alpha}}\wedge \frac{1}{n^{2\alpha-1}}\right) \right).
\end{align}
\item $0<\alpha<1/2$:
\begin{align}
& \sup_{P \in \cM_S} \bE_P \left( F_\alpha(P_n) - F_\alpha(P) \right)^2 \nonumber \\
& \quad \leq \left(\frac{3S^{1-\alpha/2}}{2n^{\alpha/2}} \wedge  \frac{5S}{2n^{\alpha}} \right)^2 \nonumber \\
& \qquad + \left( \frac{10S}{n^{2\alpha}}+ \frac{120}{\alpha^2}\left(\frac{S}{n^{2\alpha}}\wedge \frac{1}{n^{2\alpha-1}}\right) \right).
\end{align}
\end{enumerate}
Moreover, in all the bounds presented above, the first term bounds the square of the bias, and the second term bounds the variance.
\end{theorem}

\begin{theorem}[Lower bounds]\label{thm.mainlower}
We have the following lower bounds on the worst case squared error risk of MLE in estimating $F_\alpha(P)$:
\begin{enumerate}
\item $\alpha\geq 3/2$: there exists a constant $C_\alpha>0$ such that for all $n$,
\begin{align}
\sup_{P \in \cM_S} \bE_P \left( F_\alpha(P_n) - F_\alpha(P) \right)^2  \geq \frac{C_\alpha}{n}.
\end{align}
\item $1<\alpha<3/2$: if $S = cn$, for any $c>0$, then
\begin{align}
\liminf_{n\to \infty} n^{2(\alpha-1)} \cdot \sup_{P \in \mathcal{M}_S} \bE_P(F_\alpha(P_n) - F_\alpha(P))^2>0.
\end{align}
\item $1/2 \leq \alpha<1$: if $n\geq S$, then
\begin{align}
& \sup_{P \in \mathcal{M}_S} \bE_P \left( F_\alpha(P_n) - F_\alpha(P) \right)^2  \nonumber \\
& \quad \geq \frac{\alpha^2(1-\alpha)^2}{72 n^{2\alpha}} (S-1)^2 \left(1-\frac{1}{n}\right)^2 \nonumber \\
& \quad \quad + \Bigg(\frac{\alpha^2}{64en}\Bigg[(2(S-1))^{1-\alpha}-2^{-\alpha} \nonumber \\
& \qquad -\frac{1-\alpha}{4n}\left((2(S-1))^{1-\alpha}+2^{-\alpha}\right)\Bigg]^2 \nonumber \\
& \qquad - \frac{1}{2} e^{-n/4}S^{2(1-\alpha)} \Bigg),
\end{align}
\item $0<\alpha<1/2$: if $n\geq S$, then
\begin{align}
& \sup_{P \in \mathcal{M}_S} \bE_P \left( F_\alpha(P_n) - F_\alpha(P) \right)^2 \nonumber \\
& \quad \geq \frac{\alpha^2(1-\alpha)^2}{36 n^{2\alpha}} (S-1)^2 \left(1-\frac{1}{n}\right)^2.
\end{align}
\end{enumerate}
\end{theorem}

There are several interesting implications of this result, highlighted in the following corollaries.
\begin{corollary}\label{cor.largerthanone}
For any fixed $\alpha>1$, there exist universal convergence rates for $F_\alpha(P)$:
\begin{align}
& \sup_{S \in \mathbb{N}_+} \sup_{P \in \cM_S} \bE_P \left( F_\alpha(P_n) - F_\alpha(P) \right)^2 \nonumber \\
& \quad \asymp \begin{cases} n^{-2(\alpha-1)} & 1< \alpha<3/2 \\ n^{-1} & \alpha\geq 3/2 \end{cases}
\end{align}
\end{corollary}

Corollary~\ref{cor.largerthanone} implies that, when $\alpha\geq 3/2$, estimation of $F_\alpha(P)$ is extremely simple in terms of convergence rate: plug-in estimation achieves the best possible rate $n^{-1}$ (as shown in the theory of regular statistical experiments of classical asymptotic theory, see~\cite[Chap. 1.7.]{Ibragimov--Hasminskii1981}). Results of this form have appeared in the literature, for example, Antos and Kontoyiannis \cite{Antos--Kontoyiannis2001convergence} showed that it suffices to take $n \gg 1$ samples to consistently estimate $F_\alpha(P), \alpha\geq 2, \alpha \in \mathbb{Z}$. However, when $1<\alpha<3/2$, the rate $n^{-2(\alpha-1)}$ is considerably slower. Interestingly, there exist estimators that demonstrate better convergence rates for estimating $F_\alpha(P), 1<\alpha<3/2$. Jiao \emph{et al.}~\cite{Jiao--Venkat--Han--Weissman2015minimax} showed that the minimax rate in estimating $F_\alpha(P), 1<\alpha<3/2$, is $(n\ln n)^{-2(\alpha-1)}$ as long as $S\gtrsim n\ln n$, which is achieved using the general methodology developed therein for constructing minimax rate-optimal estimators for nonsmooth functionals. 

Let us now examine the case $0 < \alpha <  1$, another interesting regime that has not been characterized before. In this regime, we observe significant increase in the difficulty of the estimation problem. In particular, the relative scaling between the number of observations $n$ and the alphabet size $S$ for consistent estimation of $F_\alpha(P)$ exhibits a phase transition, encapsulated in the following.
\begin{corollary}\label{cor.l2ratesfalphalessone}
Fix $\alpha \in (0,1)$. The worst case squared error risk of the MLE $F_\alpha(P_n)$ in estimating $F_\alpha(P)$ is characterized as follows when $n\geq S$:
\begin{align}
& \sup_{P \in \mathcal{M}_S} \bE_P \left( F_\alpha(P_n) - F_\alpha(P) \right)^2 \nonumber \\
& \quad \asymp \begin{cases} \frac{S^2}{n^{2\alpha}} + \frac{S^{2-2\alpha}}{n} & 1/2< \alpha<1 \\ \frac{S^2}{n^{2\alpha}} & 0<\alpha \leq 1/2 \end{cases}
\end{align}
\end{corollary}

Corollary~\ref{cor.l2ratesfalphalessone} follows directly from Theorem~\ref{thm.main} and Theorem~\ref{thm.mainlower}. In particular, it implies that it is necessary and sufficient to take $n \gg S^{1/\alpha}$ samples to consistently estimate $F_\alpha(P), 0<\alpha<1$ using MLE. Thus, as one might expect, the scale of the number of measurements required for consistent estimation increases as $\alpha$ decreases. When $\alpha \to 0$, the number of samples required for the MLE grows super-polynomially in $S$, which is consistent with the intuition that $F_\alpha(P),\alpha \to 0$ is essentially equivalent to the alphabet size of a distribution, whose estimation is known to be very hard when there may exist symbols with very small probabilities~\cite{Efron--Thisted1976}.

We exhibit some of our findings by plotting the value required of $\ln n/\ln S$ for consistent estimation of $F_\alpha(P)$ using the MLE $F_\alpha(P_n)$, as a function of $\alpha$,  in Figure~\ref{fig.falphaphase}.

\begin{center}
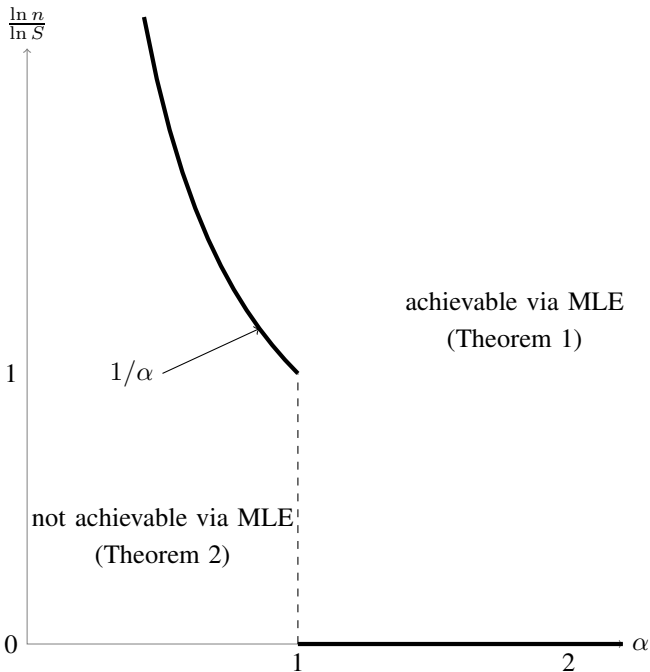

  \centering
\begin{tikzpicture}[xscale=3.6,yscale=3.6]
\draw [<->, help lines] (0.6,2.8) -- (0.6,0.6) -- (2.8,0.6);
\draw [ultra thick] (1.0316, 2.9171) -- (1.0789, 2.6879) -- (1.1263, 2.5000 ) -- (1.1737, 2.3431 )-- (1.2211, 2.2102) -- (1.2684,2.0961 ) -- (1.3158,1.9971 ) -- (1.3632,1.9103 ) -- (1.4105,1.8338 ) -- (1.4579,1.7656 ) -- (1.5053, 1.7047) -- (1.5526,1.6497 ) -- (1.6000, 1.6);
\draw [ultra thick](1.6, 0.6) -- (2.8, 0.6);
\draw [dashed] (1.6, 0.6) -- (1.6,1.6);
\node [above] at (1.1,1) {not achievable via MLE};
\node [below] at (1.1,1) {(Theorem~\ref{thm.mainlower})};
\node [left] at (0.6, 1.6) {1};
\node [left] at (0.6,0.6) {0};
\node [below] at (1.6, 0.6) {1};
\node [below] at (2.6,0.6) {2};
\node [right] at (2.8,0.6) {$\alpha$};
\node [above] at (0.6,2.8) {$\frac{\ln n}{\ln S}$};
\draw [->] (1.1,1.6) -- (1.4579,1.7656 );
\node [left] at (1.1, 1.6) {$1/\alpha$};
\node [above] at (2.4, 1.8) {achievable via MLE};
\node [below] at (2.4, 1.8) {(Theorem~\ref{thm.main})};
\end{tikzpicture}
\captionof{figure}{For any fixed point above the thick curve, consistent estimation of $F_\alpha(P)$ is achieved using MLE $F_\alpha(P_n)$ as shown in Theorem~\ref{thm.main}. For any fixed point below the thick curve in the regime $0<\alpha<1$, Theorem~\ref{thm.mainlower} shows that the MLE does not have vanishing maximum squared error risk. }
\label{fig.falphaphase}
\end{center}

It turns out that one can construct estimators that are better than the MLE in terms of required sample complexity for consistent estimation for the regime $0<\alpha<1$. Indeed, Jiao \emph{et al.}~\cite{Jiao--Venkat--Han--Weissman2015minimax} showed that the minimax rate-optimal estimator requires $n\gg \frac{S^{\frac{1}{\alpha}}}{\ln S}$ samples to achieve consistency, which attains a logarithmic improvement in the sample complexity over the MLE. 

\subsection{Estimating $H(P)$}

We not only consider $H(P_n)$, but also the so-called Miller--Madow bias-corrected estimator \cite{Miller1955} defined as
\begin{equation}
H^{\text{MM}}(P_n) = H(P_n) + \frac{S-1}{2n}. 
\end{equation}

\begin{theorem}\label{thm.entropyrisk}
The worst case squared error risk of $H(P_n)$ admits the following upper bound for all $S,n$: 
\begin{align}
& \sup_{P \in \cM_S} \bE_P \left( H(P_n) - H(P) \right)^2 \nonumber \\
& \quad \leq  \left( \ln \left( 1+ \frac{S-1}{n} \right) \right)^2 + \left(\frac{(\ln n)^2}{n}\wedge \frac{2(\ln S + 3)^2}{n}\right).
\end{align}
If $n \geq 15S$, then
\begin{align}
& \sup_{P \in \mathcal{M}_S} \bE_P \left( H(P_n) - H(P) \right)^2 \nonumber \\
& \quad \geq \frac{1}{2}\left(\frac{S-1}{2n} + \frac{S^2}{20n^2} - \frac{1}{12n^2}\right)^2 + c \frac{\ln^2 S}{n}.
\end{align}
Moreover, if $n \geq 15S$, the Miller--Madow bias-corrected estimator satisfies
\begin{align}
& \sup_{P \in \mathcal{M}_S} \bE_P \left( H^{\mathrm{MM}}(P_n) - H(P) \right)^2 \nonumber \\
& \quad \geq \frac{1}{2}\left(\frac{S^2}{20n^2} - \frac{1}{12n^2}\right)^2 + c \frac{\ln^2 S}{n},
\end{align}
where the positive constant $c>0$ in both expressions does not depend on $S$ or $n$.
\end{theorem}

Theorem~\ref{thm.entropyrisk} implies the following corollary.
\begin{corollary}\label{cor.entropyl2rate}
The worst case squared error risk of the MLE $H(P_n)$ in estimating $H(P)$ is characterized as follows when $n\geq 15S$:
\begin{equation}
\sup_{P \in \mathcal{M}_S} \bE_P \left( H(P_n) - H(P) \right)^2 \asymp \frac{S^2}{n^2} + \frac{\ln^2 S}{n}.
\end{equation}
Here the first term corresponds to the squared bias, and the second term corresponds to the variance. 
\end{corollary}

Paninski~\cite{Paninski2003} showed that if $n = cS$, where $c>0$ is a constant, the maximum squared error risk of $H(P_n)$, and the Miller--Madow bias-corrected estimator $H^{\mathrm{MM}}(P_n)$, would be bounded from zero. Paninski~\cite{Paninski2003} also showed that when $n \gg S, n \to \infty$, the MLE is consistent for estimating entropy. Corollary~\ref{cor.entropyl2rate} implies that it is necessary and sufficient to take $n\gg S$ samples for the MLE to be consistent for estimating entropy. Comparing the results for $H(P)$ with those for $F_\alpha(P)$, we see that the intuition that $H(P)$ being viewed close to $F_\alpha(P)$ when $\alpha \to 1^{-1}$ is indeed approximately correct as $H(P)$ coincides with $\alpha \to 1^{-}$ on the phase transition curve shown in Figure \ref{fig.falphaphase}.

Table~\ref{table.summary} summarizes the minimax squared error rates and the worst case squared error rates of the MLE in estimating $H(P)$ and $F_\alpha(P), \alpha>0$. It is clear that the MLE cannot achieve the minimax rates for estimation of $H(P)$, and $F_\alpha(P)$ when $0 < \alpha < 3/2$. In these cases, there exist strictly better estimators whose performance with $n$ samples is roughly the same as that of the MLE with $n \ln n$ samples. This phenomenon was termed \emph{effective sample size enlargement} in~\cite{Jiao--Venkat--Han--Weissman2015minimax}. 

\begin{table*}[t]
\begin{center}
    \begin{tabular}{| l | l | l |}
    \hline
    & Minimax squared error rates & Maximum squared error rates of MLE\\ \hline
$H(P)$    & $\frac{S^2}{(n \ln n)^2} + \frac{\ln^2 S}{n} \quad \left( n \gtrsim S/\ln S \right)$ (\cite{Jiao--Venkat--Han--Weissman2015minimax,Wu--Yang2014minimax,Valiant--Valiant2011power,Valiant--Valiant2011}) & $\frac{S^2}{n^2} + \frac{\ln^2 S}{n}\quad \left( n \gtrsim S \right)$ (Corollary~\ref{cor.entropyl2rate})  \\ \hline

$F_\alpha(P), 0<\alpha \leq \frac{1}{2}$ &  $\frac{S^2}{(n \ln n)^{2\alpha}} \quad \left( n \gtrsim S^{1/\alpha}/\ln S, \ln n \lesssim \ln S \right)$ (\cite{Jiao--Venkat--Han--Weissman2015minimax})  & $\frac{S^2}{n^{2\alpha}}\quad \left( n \gtrsim S^{1/\alpha} \right)$ (Corollary~\ref{cor.l2ratesfalphalessone})     \\
    \hline

$F_\alpha(P), \frac{1}{2}< \alpha<1$ &    $\frac{S^2}{(n \ln n)^{2\alpha}} + \frac{S^{2-2\alpha}}{n}\quad \left( n \gtrsim S^{1/\alpha}/\ln S \right)$ (\cite{Jiao--Venkat--Han--Weissman2015minimax}) & $\frac{S^2}{n^{2\alpha}} + \frac{S^{2-2\alpha}}{n}\quad \left( n \gtrsim S^{1/\alpha} \right)$ (Corollary~\ref{cor.l2ratesfalphalessone})  \\ \hline

$F_\alpha(P), 1< \alpha<\frac{3}{2}$ & $(n \ln n)^{-2(\alpha-1)}\quad \left(S \gtrsim n\ln n \right)$ (\cite{Jiao--Venkat--Han--Weissman2015minimax})  & $n^{-2(\alpha-1)}\quad \left(S \gtrsim n \right)$ (Corollary~\ref{cor.largerthanone})  \\ \hline

$F_\alpha(P), \alpha\geq \frac{3}{2}$ & $n^{-1}$ (Theorem~\ref{thm.main})  & $n^{-1}$  \\ \hline
    \end{tabular}

    \caption{Summary of results in this paper and the companion~\cite{Jiao--Venkat--Han--Weissman2015minimax}}
     \label{table.summary}
\end{center}
\end{table*}

\subsection{Dirichlet prior techniques applying to entropy estimation}

For symmetry, we restrict attention to the case where the parameter $\boldsymbol\alpha$ in the Dirichlet distribution takes the form $(a,a,\ldots,a)$. 

In comparison to MLE $H(P_n)$, where $P_n$ is the empirical distribution, the Dirichlet smoothing scheme $H(\hat{P}_B)$ has a disadvantage: it requires the knowledge of the alphabet size $S$ in general. We define
\begin{equation}
\hat{p}_{B,i} = \frac{n \hat{p}_i + a}{n + Sa},
\end{equation}
and
\begin{equation}
p_{B,i} = \bE[\hat{p}_{B,i}] = \frac{np_i + a}{n+Sa}.
\end{equation}
It is clear that 
\begin{align}
\hat{P}_B & = \frac{n}{n+Sa} P_n + \frac{Sa}{n+Sa} U_S
\label{eqn.Pbconstruction} \\
P_B & = \frac{n}{n+Sa} P + \frac{Sa}{n+Sa} U_S,
\end{align}
where $P_n$ stands for the empirical distribution, $P$ is the true distribution, and $U_S$ denotes the uniform distribution on the same alphabet with size $S$. 

\begin{theorem}\label{thm.upperbound}
If $n\geq \max\{Sa,2ea\}$, then the maximum squared error risk of $H(\hat{P}_B)$ in estimating $H(P)$ is upper bounded as
\begin{align}
& \sup_{P \in \mathcal{M}_S} \bE_P \left( H(\hat{P}_B) - H(P)\right)^2 \nonumber \\
& \quad \leq \left( \ln\left(1+\frac{S-1}{n+Sa}\right) \vee \frac{2Sa}{n+Sa} \ln \left( \frac{n+Sa}{2a}\right) \right)^2 \nonumber \\
& \qquad + \frac{2n}{(n+Sa)^2}\left[3+\ln\left(\frac{n+Sa}{a+1}\wedge S\right)\right]^2. 
\end{align}
Here the first term bounds the squared bias, and the second term bounds the variance. 
\end{theorem}

\begin{theorem}\label{thm.lowerbound}
If $n\ge\max\{15S,Sa, 2ea\}$, then the maximum $L_2$ risk of $H(\hat{P}_B)$ in estimating $H(P)$ is lower bounded as
\begin{align}
& \sup_{P \in \mathcal{M}_S} \bE_P \left( H(\hat{P}_B) - H(P)\right)^2 \nonumber \\
& \quad \geq \frac{1}{2}\left[\frac{(S-1)a}{4(n+Sa)}\ln \left( \frac{n+Sa}{a}\right) + \frac{S-1}{8n} + \frac{S^2}{80n^2} - \frac{1}{48n^2}\right]^2 \nonumber \\
& \qquad + c \frac{\ln^2 S}{n},
\end{align}
where $c>0$ is a universal constant that does not depend on $a, S$, or $n$.  

If $n<Sa$, then we have
\begin{equation}
\sup_{P \in \mathcal{M}_S} \bE_P \left( H(\hat{P}_B) - H(P)\right)^2 \geq \left( \frac{S-1}{2S} \right)^2 \ln^2 S.
\end{equation}

If $n<2ea$, then we have
\begin{equation}
\sup_{P \in \mathcal{M}_S} \bE_P \left( H(\hat{P}_B) - H(P)\right)^2 \geq  \left( \frac{S-1}{S+2e} \right)^2 \ln^2 S.
\end{equation}

If $n<15S, n\geq 2ea$, then we have
\begin{align}
& \sup_{P \in \mathcal{M}_S} \bE_P \left( H(\hat{P}_B) - H(P)\right)^2 \nonumber \\
& \quad \geq \left[\left( \frac{(S-1)a}{4(n+Sa)}\ln  \left( \frac{n+Sa}{a}\right) + \frac{\lfloor n/15 \rfloor}{8n} - \frac{1}{16n} \right)_+ \right]^2,
\end{align}
  where $\lfloor x \rfloor$ is the largest integer that does not exceed $x$, and $(x)_+ = \max\{x,0\}$ represents the positive part of $x$. 
\end{theorem}

The following corollary immediately follows from Theorem~\ref{thm.upperbound} and Theorem~\ref{thm.lowerbound}. 
\begin{corollary}
If $n\gg S$ and $a$ is upper bounded by a constant, then the maximum squared error risk of $H(\hat{P}_B)$ vanishes. Conversely, if $n \lesssim S$, then the maximum squared error risk of $H(\hat{P}_B)$ is bounded away from zero. 
\end{corollary}

The next theorem presents a lower bound on the maximum risk of the Bayes estimator under Dirichlet prior. Since we have assumed that all $\alpha_i = a,1\leq i\leq S$, the Bayes estimator under Dirichlet prior is
\begin{equation}
\hat{H}^{\mathsf{Bayes}} = \psi(Sa +n + 1) - \sum_{i = 1}^S \frac{a + X_i}{Sa + n} \psi(a + X_i + 1). 
\end{equation}

\begin{theorem}\label{thm.bayeslowerbound}
If $S \geq 2(n+1)$, then
\begin{equation}
\sup_{P\in \mathcal{M}_S} \bE_P \left( \hat{H}^{\mathsf{Bayes}} - H(P) \right)^2 \geq \left(  \ln \left( \frac{Sa + S/2}{Sa + n + e^{-\gamma}} \right) \right)^2,
\end{equation}
where $\gamma \approx 0.5772$ is the Euler-–Mascheroni constant. 
\end{theorem}

Evident from Theorem~\ref{thm.upperbound},~\ref{thm.lowerbound}, and~\ref{thm.bayeslowerbound} is the fact that in the best situation (i.e. $a$ not too large), both the Dirichlet prior smoothed plug-in estimator and the Bayes estimator under Dirichlet prior still require at least $n\gg S$ samples to be consistent, which is the same as MLE. In contrast, the estimators in Valiant and Valiant~\cite{Valiant--Valiant2011,Valiant--Valiant2011power,Valiant--Valiant2013estimating}, Jiao \emph{et al.}~\cite{Jiao--Venkat--Han--Weissman2015minimax}, Wu and Yang~\cite{Wu--Yang2014minimax} are consistent if $n\gg \frac{S}{\ln S}$, which is the optimal sample complexity. Thus, we can conclude that the Dirichlet smoothing technique does \emph{not} solve the entropy estimation problem. 

\section{Fundamental ideas of our analysis}\label{sec.ideas}

In this section, we discuss the fundamental tools we employed to obtain the results in Section~\ref{sec.mainresults}, as well as general recipes we suggest for analyzing performances of functional estimators. 

\subsection{Variance}

The variance characterizes the degree to which the random variable $F(\hat{P})$ is fluctuating around its expectation, and the field of concentration inequalities perfectly fits our glove to give the desired results. For all the functionals we consider, it turns out that the Efron--Stein inequality~\cite{Efron--Stein1981jackknife} and the bounded differences inequality give very tight bounds. For completeness we state them below. 

\begin{lemma}\label{lemma.esnew}\cite[Efron--Stein inequality, Theorem 3.1]{Boucheron--Lugosi--Massart2013}
Let $Z_1,\ldots,Z_n$ be independent random variables and let $f(Z_1,Z_2,\ldots,Z_n)$ be a square integrable function. Moreover, if $Z_1',Z_2',\ldots,Z_n'$ are independent copies of $Z_1,Z_2,\ldots,Z_n$ and if we define, for every $i = 1,2,\ldots,n$,
\begin{equation}
f_i' = f(Z_1,Z_2,\ldots,Z_{i-1},Z_i',Z_{i+1},\ldots,Z_n),
\end{equation}
then
\begin{equation}
\mathsf{Var}(f) \leq \frac{1}{2} \sum_{i = 1}^n \bE \left[ (f - f_i')^2 \right].
\end{equation}
\end{lemma}

The following inequality, which is called the bounded differences inequality, is a useful corollary of the Efron--Stein inequality.
\begin{lemma}\label{lemma.es}\cite[Bounded differences inequality, Corollary 3.2]{Boucheron--Lugosi--Massart2013}
If function $f \colon \mathcal{Z}^n \to \mathbb{R}$ has the \emph{bounded differences property}, i.e., for some nonnegative constants $c_1,c_2,\ldots,c_n$,
\begin{align}
& \sup_{z_1,\ldots,z_n,z_i'\in \mathcal{Z}} |f(z_1,\ldots,z_n) - f(z_1,\ldots,z_{i-1},z_i',z_{i+1},\ldots,z_n)| \nonumber \\
& \quad \leq c_i,
\end{align}
for every $1\leq i \leq n$, then
\begin{equation}
\mathsf{Var}(f(Z_1,Z_2,\ldots,Z_n)) \leq \frac{1}{4}\sum_{i=1}^n c_i^2,
\end{equation}
given that $Z_1,Z_2,\ldots,Z_n$ are independent random variables.
\end{lemma}

We refer the readers to Boucheron \emph{et al.}~\cite{Boucheron--Lugosi--Massart2013} for a modern exposition of the concentration inequality toolbox. 

\subsection{Bias} \label{sec.biasidea}

It turns out that the bias analysis in estimation, albeit widely studied in statistics, seems to still largely bear an asymptotic and expansion nature in the mainstream statistical literature~\cite{Barndorff--Cox1989asymptotic,Small2010expansions}. In particular, the bootstrap~\cite{Efron1979bootstrap} as a method for estimating functionals was essentially only analyzed in an asymptotic setting~\cite{Hall1992bootstrap}. Among asymptotic analysis techniques, probably the most popular one is the Taylor expansion. We will show that the Taylor expansion may encounter great difficulties in analyzing the bias of MLE in information measure estimation. Then, we will introduce the field of \emph{approximation theory using positive linear operators} and demonstrate that it is essentially equivalent to \emph{nonasymptotic} bias analysis for plug-in functional estimators. In doing so, we present the readers with abundant handy tools from approximation theory, which could be readily applicable to many problems that may seem highly intractable with standard expansion methods. 

We start from entropy estimation. In the literature, considerable effort has been devoted to understanding the non-asymptotic performance of the MLE $H(P_n)$ in estimating $H(P)$. One of the earliest investigations in this direction is due to Miller~\cite{Miller1955} in 1955, who showed that, for any fixed distribution $P$,
\begin{equation}\label{eqn.miller}
\bE H(P_n) = H(P) - \frac{S-1}{2n} + O\left(\frac{1}{n^2}\right).
\end{equation}
Equation~(\ref{eqn.miller}) was later refined by Harris~\cite{Harris1975} using higher order Taylor series expansions to yield
\begin{equation}\label{eqn.harris}
\bE H(P_n) = H(P) - \frac{S-1}{2n} + \frac{1}{12 n^2} \left( 1- \sum_{i = 1}^S \frac{1}{p_i} \right) + O\left( \frac{1}{n^3} \right).
\end{equation}
Harris's result reveals an undesirable consequence of the Taylor expansion method: one cannot obtain uniform bounds on the bias of the MLE. Indeed, the term $\sum_{i = 1}^S \frac{1}{p_i}$ can be arbitrarily large for some distribution $P$. However, it is evident that both $H(P_n)$ and $H(P)$ are bounded above by $\ln S$, since the maximum entropy of any distribution supported on $S$ elements is $\ln S$. Conceivably, for such a distribution $P$ that would make $\sum_{i = 1}^S \frac{1}{p_i}$ very large, we need to compute even higher order Taylor expansions to obtain more accuracy, but even with such efforts we cannot obtain a uniform bias bound for all $P$.

We gain one of our key insights into the bias of the MLE by relating it to the approximation error induced by the {\em Bernstein polynomial approximation} of the function $f$, which was first observed in Paninski~\cite{Paninski2003}. To see this, we first compute the bias of $F(P_n)$ in estimating the functional $F(P)$ in (\ref{eqn.generalf}). 

\begin{lemma}\label{lemma.biasgeneralf}
The bias of the estimator $F(P_n)$ is given by
\begin{align}\label{eqn.biasgeneralfequation}
\mathsf{Bias}(F(P_n)) &\triangleq \bE F(P_n) - F(P) \nonumber \\
&= \sum_{i = 1}^S \left( \sum_{j = 0}^n f \left( \frac{j}{n} \right) \binom{n}{j} p_i^j(1-p_i)^{n-j} - f(p_i) \right).
\end{align}
\end{lemma}

The bias term in (\ref{eqn.biasgeneralfequation}) can be equivalently expressed as\footnote{In the literature of combinatorics, the sum $\sum_{j = 0}^n a_{j,n}B_{j,n}(x)$ is called the Bernoulli sum, and various approaches have been proposed to evaluate its asymptotics \cite{Jacquet--Szpankowski1999entropy}, \cite{Flajolet1999singularity}, \cite{Cichon--Golkbiewski--Kardas--Klonowski}. }
\begin{align}
\mathsf{Bias}(F(P_n)) & = \sum_{i = 1}^S  \left( \sum_{j = 0}^n f\left( \frac{j}{n} \right) B_{j,n}(p_i) - f(p_i) \right) \\
& = \sum_{i = 1}^S \left( B_n[f](p_i) - f(p_i)\right),
\end{align}
where $B_{j,n}(x) \triangleq \binom{n}{j} x^j(1-x)^{n-j}$ is the well-known Bernstein polynomial basis, and $B_n[f](x)$ is the so-called Bernstein polynomial for function $f(x)$.  Bernstein in 1912~\cite{Bernstein1958collected} provided an insightful constructive proof of the Weierstrass theorem on approximation of continuous functions using polynomials, by showing that the Bernstein polynomial of any continuous function converges uniformly to that function. From a functional analytic viewpoint, the Bernstein polynomial is an operator that maps a continuous function $f\in C[0,1]$ to another continuous function $B_n[f] \in C[0,1]$. This operator is linear in $f$, and is \emph{positive} because $B_n[f]$ is also pointwise non-negative if $f$ is pointwise non-negative. Evidently, bounding the approximation error incurred by the Bernstein polynomial is equivalent to bounding the bias of the MLE $f(X/n)$, where $X \sim \mathsf{B}(n,x)$. Fortunately, the theory of \emph{approximation using positive linear operators}~\cite{Paltanea2004} provides us with advanced tools that are very effective for the bias analysis our problem calls for. A century ago, probability theory served Bernstein in breaking new ground in function approximation. It is therefore very satisfying that advancements in the latter have come full circle to help us better understand probability theory and statistics. We briefly review the general theory of approximation using positive linear operators below.

\subsubsection{Approximation theory using positive linear operators}

Generally speaking, for any estimator $\hat{\theta}$ of a parametric model indexed by $\theta$, the expectation $f \mapsto \bE_\theta f(\hat{\theta})$ is a positive linear operator for $f$, and analyzing the bias $\bE_\theta f(\hat{\theta}) -f(\theta)$ is equivalent to analyzing the approximation properties of the positive linear operator $\bE_\theta f(\hat{\theta})$ in approximating $f(\theta)$. Hence, analyzing the bias of \emph{any} plug-in estimator for functionals of parameters from \emph{any} parametric families can be recast as a problem of approximation theory using positive linear operators~\cite{Paltanea2004}.

Conversely, given a positive linear operator $L(f)(x)$ that operates on the space of continuous functions, the Riesz--Markov--Kakutani theorem implies that under mild conditions the operator may be written as 
\begin{align}
L(f)(x) = \int_I f d\mu_x = \mathbb{E}_{\mu_x} f(Z), Z\sim \mu_x,
\end{align}
where $\{\mu_x\}$ is a set of probability measures parametrized by $x$, which may be viewed as a parameter. If we view the random variable $Z$ as a summary statistics to plug-in the functional $f(\cdot)$, the positive linear operator $L(f)(x)$ is nothing but the expectation of the plug-in estimator $f(Z)$. In this sense, there exists a one-to-one correspondence between essentially the most general bias analysis problem in statistics, and the most general positive linear operator approximation problem in approximation theory. 

After more than a century's active research on approximation using positive linear operators, we now have highly non-trivial tools for positive linear operators of functions on one dimensional compact sets, but the general theory for vector valued multivariate functions on non-compact sets is still far from complete~\cite{Paltanea2004}. In the next subsection, we present a sample of existing results in approximation using positive linear operators, corollaries of which will be used to analyze the bias of the MLE for two examples: $F_\alpha(P)$ and $H(P)$. 

\subsubsection{Some general results in bias analysis}\label{sec.generalbiasanalysis}

First, some elementary approximation theoretic concepts need to be introduced in order to characterize the degree of \emph{smoothness} of functions. For $I\subset \mathbb{R}$ an interval, the first-order modulus of smoothness $\omega^1(f,t),t\geq 0$ is defined as~\cite{Paltanea2004}
\begin{equation}
\omega^1(f,t) \triangleq \sup \{ |f(u) - f(v)|: u,v \in I, |u-v|\leq t \}.
\end{equation}

The second-order modulus of smoothness $\omega^2(f,t),t\geq 0$ \cite{Paltanea2004} is defined as
\begin{align}
\omega^2(f,t) & \triangleq \sup \Bigg \{ \left| f(u) - 2f\left(\frac{u+v}{2}\right) + f(v) \right|\colon \nonumber \\
& \qquad \qquad u, v\in I, |u-v|\leq 2t \Bigg \}.
\end{align}

Ditzian and Totik~\cite{Ditzian--Totik1987} introduced a class of moduli of smoothness, which proves to be extremely useful in characterizing the incurred approximation errors. For simplicity, for functions defined on $[0,1]$, $\varphi(x) = \sqrt{x(1-x)}$, the first-order Ditzian--Totik modulus of smoothness is defined as
\begin{align}
\omega^1_\varphi(f,t) & \triangleq \sup \Bigg \{ |f(u) - f(v)|\colon \nonumber \\
& \qquad \qquad u,v\in [0,1], |u-v| \leq t \varphi\left( \frac{u+v}{2}\right) \Bigg \},
\end{align}
and the second-order Ditzian--Totik modulus of smoothness is defined as
\begin{align}
\omega_\varphi^2(f,t) &  \triangleq \sup \Bigg \{ \left | f(u) - 2f \left( \frac{u+v}{2}\right) +f(v) \right|\colon \nonumber \\
&\qquad \qquad  u, v\in [0,1], |u-v| \leq 2t \varphi\left( \frac{u+v}{2} \right) \Bigg \}  .
\end{align}

Recall that we denote by $e_j,j\in \mathbb{N}_+ \cup \{0\}$, the monomial functions $e_j(y) = y^j,y\in I$. The first estimate for general positive linear operators, using modulus $\omega^2$ and with precise constants, was given by Gonska~\cite{Gonska1979quantitative}. We rephrase Paltanea~\cite[Cor. 2.2.1.]{Paltanea2004} as follows. Note that notation $e_1 - x e_0$ denotes a continuous function on $I$ which is the difference of a linear function $y$ and a constant function with constant value $x$ over $I$. In other words, it is an abbreviation of $e_1(y) - x e_0(y),y\in I$, which is a function of $y$ rather than $x$. 

For a positive linear functional $F$, we adopt the following notation
  \begin{align} \label{eqn.functionalbiasandvariance}
    B_F(x) = \left|F(e_1)-xF(e_0)\right|,\quad V_F = F\left((e_1-F(e_1)e_0)^2\right),
  \end{align}
  which represent the ``bias'' and ``variance'' of a positive linear functional $F$.

\begin{lemma}\cite[Cor. 2.2.1.]{Paltanea2004}\label{lemma.paltaneacor221}
Let $F\colon C(I) \to \mathbb{R}$ be a positive linear functional, where $I \subset \mathbb{R}$ is an interval. Suppose that $F(e_0) = 1, t>0, \mathrm{length}(I)\geq 2t, s\geq 2$. Then,
\begin{align}
|F(f) - f(x)| & \leq B_F(x)\frac{\omega^1(f,t)}{t} \nonumber \\
& \qquad + \left( 1 + \frac{F(|e_1 - xe_0|^s)}{2t^s} \right)\omega^2(f,t).
\end{align}
\end{lemma}

We remark that Lemma~\ref{lemma.paltaneacor221} can be applied to bound the bias of plug-in estimators in very general models. For example, consider an arbitrary statistical experiment $\{P_\theta, \theta \in I\}$, from which we obtain $n$ i.i.d. samples $X_1,X_2,\ldots,X_n \sim P_\theta$. For any estimator $\hat{\theta}_n$, we would like to analyze the bias of the plug-in estimator $f(\hat{\theta}_n)$ for functional $f(\theta)$.

Suppose $\mathrm{length}(I)\geq 2t, s\geq 2$, then Lemma~\ref{lemma.paltaneacor221} implies that
\begin{align}
|\bE_\theta f(\hat{\theta}_n) - f(\theta)| & \leq |\bE_\theta \hat{\theta}_n - \theta| \frac{\omega^1(f,t)}{t} \nonumber \\
& \qquad + \left( 1+ \frac{\bE |\hat{\theta}_n - \theta|^s}{2t^s} \right)\omega^2(f,t).
\end{align}

If we further assume that $\hat{\theta}_n$ is an unbiased estimator for $\theta$, i.e., $\bE_\theta \hat{\theta}_n = \theta$ holds for all $\theta \in I$, then we have
\begin{equation}
|\bE_\theta f(\hat{\theta}_n) -f(\theta)| \leq \left( 1+ \frac{\bE |\hat{\theta}_n - \theta|^s}{2t^s} \right)\omega^2(f,t).
\end{equation}

Taking $s = 2$ and assuming $\mathsf{Var}(\hat{\theta}_n) \leq \mathrm{length}(I)/2$, we have
\begin{equation}
|\bE_\theta f(\hat{\theta}_n) - f(\theta)| \leq \frac{3}{2} \omega^2(f, \sqrt{\mathsf{Var}(\hat{\theta}_n)}),
\end{equation}
after we take $t = \sqrt{\bE |\hat{\theta}_n - \theta|^2}$.

We remark that Lemma~\ref{lemma.paltaneacor221} is only one way to analyze the bias, which is by no means always tight. For example, the following estimate using Ditzian--Totik modulus is significantly better than Lemma~\ref{lemma.paltaneacor221} for certain functions such as the entropy.

\begin{lemma}\cite[Thm. 2.5.1.]{Paltanea2004}\label{lemma.paltaneathm251}
  If $F\colon C[0,1]\to \mathbb{R}$ is a linear positive functional and $F(e_0)=1$, then we have
    \begin{align}\label{eq:inequality1}
      |F(f)-f(x)| \le \frac{B_F(x)}{2h_1\varphi(x)}\cdot \omega^1_\varphi(f,2h_1) + \frac{5}{2}\omega^2_\varphi(f,h_1),
    \end{align}
    for all $f\in C[0,1]$ and $0<h_1\le \frac{1}{2}$, where $\varphi(x)=\sqrt{x(1-x)}$ and $h_1=\sqrt{F\left((e_1-xe_0)^2\right)}/\varphi(x)=\sqrt{V_F+(B_F(x))^2}/\varphi(x)$. The ``bias'' $B_F(x)$ and ``variance'' $V_F(x)$ are defined in~(\ref{eqn.functionalbiasandvariance}). 
  \end{lemma}

Considering the same statistical experiment $\{P_\theta, \theta \in I\}$, and the plug-in estimator $f(\hat{\theta}_n)$ for $f(\theta)$, if $\hat{\theta}_n$ is unbiased for $\theta$ and $\mathsf{Var}(\hat{\theta}_n) \leq \frac{\varphi(\theta)^2}{4}$, then it follows from Lemma~\ref{lemma.paltaneathm251} that
\begin{equation}
|\bE_\theta f(\hat{\theta}_n) - f(\theta)| \leq \frac{5}{2} \omega_\varphi^2 \left( f, \frac{\sqrt{\mathsf{Var}(\hat{\theta}_n)}}{\varphi(\theta)} \right),
\end{equation}
after we take $t = \frac{\sqrt{\mathsf{Var}(\hat{\theta}_n)}}{\varphi(\theta)}$.

For certain functions $f(\theta)$ and statistical models Lemma~\ref{lemma.paltaneathm251} is stronger than Lemma~\ref{lemma.paltaneacor221}. For example, if $f(\theta) = -\theta \ln \theta, \theta \in [0,1]$, and we have $n\cdot\hat{\theta}_n \sim \mathsf{B}(n,\theta)$. We will show in Lemma~\ref{lemma.dtmoduluscomputation} that $\omega_\varphi^2(f,t) = \frac{t^2\ln 4}{1+t^2}$, and $\omega^2(f,t) = t\ln 4$. We also have $\mathsf{Var}(\hat{\theta}_n) = \frac{\theta(1-\theta)}{n}$. Hence, Lemma~\ref{lemma.paltaneacor221} gives the upper bound
\begin{equation}\label{eqn.pointwiseboundforentropyf}
|\bE_\theta f(\hat{\theta}_n) - f(\theta)| \leq \frac{3\ln 4}{2} \sqrt{\frac{\theta(1-\theta)}{n}},
\end{equation}
whereas Lemma~\ref{lemma.paltaneathm251} gives
\begin{equation}\label{eqn.normboundentropyf}
|\bE_\theta f(\hat{\theta}_n) - f(\theta)| \leq \frac{5\ln 4}{2n} \cdot \frac{1}{1+1/n},
\end{equation}
which is much stronger when $n$ is large and $\theta$ not too close to the endpoints of $[0,1]$.

There also exist various estimates for the bias when the parameter lies in sets other than an interval in $\mathbb{R}$. However, the bounds we presented are in general \emph{not} optimal for specific functionals, thereby leaving ample room for future development. For example, note that (\ref{eqn.pointwiseboundforentropyf}) is stronger than (\ref{eqn.normboundentropyf}) when $\theta \leq 1/n$, but Han, Jiao, and Weissman~\cite{Han--Jiao--Weissman2015adaptive} showed that when $\theta \leq 1/n$ the pointwise bound in~(\ref{eqn.pointwiseboundforentropyf}) is still strictly sub-optimal for the entropy functional. Unsurprisingly, to obtain the results in Section~\ref{sec.mainresults}, we need to go beyond the general results in approximation theory, and incorporate the structure of specific functions.

{\em Note:} In approximation theory literature, researchers have explored the interactions between general positive linear operator approximation and its probabilistic counterpart decades ago~\cite{Strukov--Timan1977mathematical,Walk1980probabilistic,Hahn1981note}. However, in statistics literature related to positive linear approximation, usually only specific operators are used, such as the Bernstein operator~\cite{Braess--Sauer2004}, and the focus may not be on obtaining the tightest bound on bias~\cite{Diaconis--Zabell1991closed,Feller2008introduction}. 

\subsection{Lower bounds}

To lower bound the worst case performance of a specific estimator, we have essentially two approaches: first, to analyze the bias or the variance of the specific estimator carefully; second, to prove a lower bound that is satisfied by all the estimators, which naturally include the specific estimator we need to analyze. These two approaches have different relative advantages and disadvantages, so we utilize them together in the lower bound construction. 

We refer the readers to Tsybakov~\cite{Tsybakov2008} for a nice collection of techniques to prove minimax lower bounds. One specific approach we use is the van Trees inequality, which we quote below.

Let $(\mathcal{X},\mathcal{F},P_\theta; \theta \in \Theta)$ be a dominated family of distributions on some sample space $\mathcal{X}$; denote the dominating measure by $\mu$. Assume $\Theta$ is a closed interval on the real line. Let $f(x|\theta)$ denote the density of $P_\theta$ with respect to $\mu$. Let $\pi$ be some probability distribution on $\Theta$ with a density $\lambda(\theta)$ with respect to Lebesgue measure. Suppose that $\lambda$ and $f(x|\cdot)$ are both absolutely continuous ($\mu$-almost surely), and that $\lambda$ converges to zero at the endpoints of the interval $\Theta$. We define 
\begin{align}
\mathcal{I}(\theta) & = \mathbb{E}_\theta \left( \frac{\partial \log f(X|\theta)}{\partial \theta} \right)^2 \\
\mathcal{I}(\lambda) & = \mathbb{E}\left( \frac{\mathrm{d \log \lambda(\theta)}}{\mathrm{d}\theta} \right)^2
\end{align}
the Fisher information for $\theta$ and for a location parameter in $\lambda$, respectively. We assume $\mathcal{I}(\theta)$ is continuous in $\theta$. We have the following inequality.
\begin{lemma}[van Trees inequality]\cite{Gill--Levit1995applications} \label{lemma.vantrees}
Under assumptions above, the average risk of an arbitrary estimator $\hat{\psi}(X)$ in estimating an absolutely continuous functional $\psi(\theta)$ under squared error loss satisfies the following inequality:
\begin{align}
\mathbb{E} \left( \hat{\psi}(X) - \psi(\theta) \right)^2 \geq \frac{(\mathbb{E} \psi'(\theta))^2}{\mathbb{E}[\mathcal{I}(\theta)] + \mathcal{I}(\lambda)} 
\end{align}
\end{lemma}

\section{Proofs of the upper bounds}\label{sec.upperboundproof}

In order to upper bound the maximum squared error risk of any estimator, a natural approach would be to analyze the squared bias term and the variance term separately. Then, it suffices to find proper tools to give \emph{nonasymptotic} analysis of the bias and variance. 

\subsection{Bounding the bias}
We first work to bound the bias. Lemma~\ref{lemma.biasgeneralf} shows that the bias of $F(P_n)$ could be represented as
\begin{equation}
\mathsf{Bias}(F(P_n)) = \sum_{i = 1}^S \left( B_n[f](p_i) - f(p_i) \right),
\end{equation}
where $B_n[f](x)$ is the Bernstein polynomial corresponding to $f(x)$. The following lemma summarizes some state-of-the-art bounds for approximation error of Bernstein polynomials. Lemma~\ref{lemma.bernsteinerror} can be derived easily from the general theory we presented in Section~\ref{sec.generalbiasanalysis}. We emphasize that one cannot expect the bounds in Lemma~\ref{lemma.bernsteinerror} to be tight for any $f\in C[0,1]$, since the Bernstein approximation error itself could be a very complicated function in $C[0,1]$, and Lemma~\ref{lemma.bernsteinerror} is using relatively simple functions to upper bound it. 

\begin{lemma}\label{lemma.bernsteinerror}
The following bounds are valid for function approximation error incurred by Bernstein polynomials:
\begin{enumerate}
\item \emph{Pointwise estimate:} \cite[Cor. 2.2.1]{Paltanea2004}\cite{Paltanea2008some} for all continuous functions $f$ on $[0,1]$,
\begin{equation}\label{eqn.bern1}
|f(x)-B_n[f](x)| \leq \frac{3}{2} \omega^2\Bigg(f, \sqrt{\frac{x(1-x)}{n}}\Bigg),
\end{equation}
and the constant $3/2$ is shown by \cite{Paltanea2008some} to be the best constant;
\item \emph{Norm estimate:} \cite[Cor. 4.1.10]{Paltanea2004} for $\varphi(x) = \sqrt{x(1-x)}$ and all continuous functions $f$ on $[0,1]$, we have
\begin{equation}\label{eqn.bern3}
\| B_n[f] - f \|_\infty \leq \frac{5}{2}\omega_\varphi^2(f,n^{-1/2});
\end{equation}
\item \cite[Eqn. 10.3.4]{Devore--Lorentz1993} for $f\in C^2[0,1]$, i.e., twice continuously differentiable,
\begin{equation}\label{eqn.bern2}
 | f(x) - B_n[f](x) | \leq \| f'' \|_\infty \frac{x(1-x)}{2n};
\end{equation}
\end{enumerate}
\end{lemma}

\begin{proof}
The pointwise estimate of Lemma~\ref{lemma.bernsteinerror} follows from Lemma~\ref{lemma.paltaneacor221}. The norm estimate of Lemma~\ref{lemma.bernsteinerror} follows from Lemma~\ref{lemma.paltaneathm251}. Regarding the third part, suppose random variable $X \sim \mathsf{B}(n,x)$. We have
\begin{align}
& |f(x) - B_n[f](x)| \nonumber \\
& \quad = |\bE_x f(X/n) - f(x)| \\
& \quad = |\bE_x [f'(x)(X/n-x) + \frac{1}{2} f''(\xi_X) (X/n-x)^2 ]| \\
&  \quad = \frac{1}{2} |\bE_x f''(\xi_X) (X/n-x)^2| \\
& \quad \leq \frac{\| f''\|_\infty}{2} |\bE_x (X/n-x)^2| \\
& \quad = \frac{\| f''\|_\infty}{2} \frac{x(1-x)}{n},
\end{align}
where we used Taylor expansion for $f(X/n)$ at point $x$ with the Lagrange remainder. The proof is complete. 
\end{proof}

\begin{remark}
Note that although (\ref{eqn.bern3}) is in the form of an upper bound, it has been shown to be a lower bound as well. Totik~\cite{Totik1994approximation} showed the following equivalence property on the norm estimate of Bernstein approximation errors
\begin{equation}\label{eqn.normtotikbound}
\| B_n[f](x) - f(x) \|_\infty \asymp \omega_\varphi^2(f,n^{-1/2}). \footnote{Note that it is a remarkable fact that (\ref{eqn.normtotikbound}) holds for any continuous function $f(x)$. The lower bound proof of (\ref{eqn.normtotikbound}) is considered one of the remarkable results in approximation theory, and currently there are no ``short'' proofs of this fact. Indeed, Ditzian~\cite[Section 8]{Ditzian2007} mentioned that \emph{``I still would like to see a new simple proof of (8.4) (Equation~(\ref{eqn.normtotikbound})) which I am sure will have implications for other operators.''} }
\end{equation}
\end{remark}

It is easy to calculate the second-order modulus of smoothness and the Ditzian--Totik second-order modulus of smoothness for functions $x^\alpha$ and $-x \ln x$. The results are presented in the following lemma.

\begin{lemma}\label{lemma.dtmoduluscomputation}
We have
\begin{center}
    \begin{tabular}{| l | l | l | l |}
    \hline
    & $x^\alpha, 0<\alpha<1$ & $x^\alpha, 1<\alpha< 2$ & $-x \ln x $ \\ \hline
    $\omega^2(f,t)$ & $|2-2^\alpha|t^\alpha$ &  $|2-2^\alpha|t^\alpha$ & $t \ln 4$  \\ \hline
    $\omega^2_\varphi(f,t)$ &  $|2-2^\alpha| \frac{t^{2\alpha}}{(1+t^2)^{\alpha}}$ & $ \asymp t^2$ & $\frac{t^2 \ln 4}{1+t^2}$ \\
    \hline
    \end{tabular}
\end{center}
where the second-order modulus results hold for $0<t\leq 1/2$, and the Ditizan--Totik second-order modulus results hold for $0<t\leq 1$.
\end{lemma}

\subsubsection{Bias of $F_\alpha(P_n)$}
We first bound the bias incurred by $F_\alpha(P_n)$. 
\begin{enumerate}
\item $\alpha \geq 2$:

In this case, $f \in C^2[0,1]$, applying the third part of Lemma~\ref{lemma.bernsteinerror}, 
\begin{equation}
|f(x) - B_n[f](x)| \leq \frac{\alpha (\alpha-1) x(1-x)}{2n}.
\end{equation}
Thus, we have
\begin{equation}
|\mathsf{Bias}(F_{\alpha}(P_n))| \leq \sum_{i =1}^S \alpha (\alpha-1) \frac{p_i(1-p_i)}{2n} \leq \frac{\alpha (\alpha-1)}{2n}.
\end{equation}

\item $1<\alpha<2$

The following lemma presents a bound on the bias of $F_\alpha(P_n)$, which does not depend on the alphabet size $S$. We note that the proof of Lemma~\ref{lemma.biasalphaonetwo} heavily utilizes the special properties of function $x^\alpha$ and the fact that $\sum_{i= 1}^S p_i =1$. 

\begin{lemma}\label{lemma.biasalphaonetwo}
The bias of $F_\alpha(P_n)$ for estimating $F_\alpha(P), 1<\alpha<2$, is upper bounded by the following:
\begin{equation}
|\mathsf{Bias}(F_{\alpha}(P_n))| \leq \frac{4}{n^{\alpha-1}}.
\end{equation}
\end{lemma}

We also present two additional bounds involving the alphabet size $S$. Using the pointwise estimate in Lemma~\ref{lemma.bernsteinerror}, the bias term of the MLE is upper bounded as follows for all $0<\alpha<2, \alpha \neq 1$:

\begin{align}
& \sum_{i = 1}^S \frac{3}{2}|2-2^\alpha| \left( \frac{p_i(1-p_i)}{n} \right)^{\alpha/2} \nonumber \\
& \quad \leq  \frac{3}{2}|2-2^\alpha|\frac{1}{n^{\alpha/2}} \sum_{i =1}^S p_i^{\alpha/2} \\
& \quad \leq \frac{3}{2}|2-2^\alpha|\frac{1}{n^{\alpha/2}} S \frac{1}{S^{\alpha/2}}  \\
& \quad = \frac{3}{2}|2-2^\alpha|\frac{S^{1-\alpha/2}}{n^{\alpha/2}}. \label{eqn.biasoneandtwopoint}
\end{align}

Using the norm estimate in Lemma~\ref{lemma.bernsteinerror}, when $1<\alpha<2$, the bias would be upper bounded by $C_{\alpha,n} \frac{5S}{2n}$, where $C_{\alpha,n} = n \omega_\varphi^2(x^\alpha,n^{-1/2})$ is a finite positive constant such that $\limsup_{n\to \infty} C_{\alpha,n}<\infty$ for $1<\alpha<2$. Combining Lemma~\ref{lemma.biasalphaonetwo}, the pointwise estimate, and the norm estimate in Lemma~\ref{lemma.bernsteinerror}, we know that the bias of $F_\alpha(P_n)$ for $1<\alpha<2$ is upper bounded as
\begin{equation}
|\mathsf{Bias}(F_\alpha(P_n))| \leq  \frac{4}{n^{\alpha-1}} \wedge  \frac{3}{2}|2-2^\alpha|\frac{S^{1-\alpha/2}}{n^{\alpha/2}} \wedge C_{\alpha,n}\frac{5S}{2n}.
\end{equation}

\item $0<\alpha<1$:

The pointwise estimate from Lemma~\ref{lemma.bernsteinerror} is worked out in~(\ref{eqn.biasoneandtwopoint}). Using the norm estimate in Lemma~\ref{lemma.bernsteinerror}, the bias would be upper bounded by $|2-2^\alpha| \frac{5S}{2n^\alpha}$. Combining the pointwise estimate and the norm estimate, we know that the bias of $F_\alpha(P_n)$ for $0<\alpha<1$ is upper bounded as
\begin{equation}
|\mathsf{Bias}(F_\alpha(P_n))| \leq   \frac{3}{2}|2-2^\alpha|\frac{S^{1-\alpha/2}}{n^{\alpha/2}} \wedge |2-2^\alpha|\frac{5S}{2n^\alpha}.
\end{equation}
\end{enumerate}

\subsubsection{Bias of $H(P_n)$}

We then bound the bias incurred by $H(P_n)$. Using the norm estimate in Lemma~\ref{lemma.bernsteinerror}, we know
\begin{equation}
|\mathsf{Bias}(H(P_n))| \leq \frac{5S \ln 4}{2n}.
\end{equation}
Using the pointwise estimate in Lemma~\ref{lemma.bernsteinerror}, we obtain
\begin{equation}
|\mathsf{Bias}(H(P_n))| \leq \frac{3}{2}\sqrt{\frac{S}{n}} \ln 4.
\end{equation}

It was shown by Paninski~\cite[Prop. 1]{Paninski2003} that the squared bias of MLE $H(P_n)$ is upper bounded as
\begin{equation}\label{eqn.paninskientropy}
(\mathsf{Bias}(H(P_n)))^2 \leq \left( \ln \left( 1+ \frac{S-1}{n} \right) \right)^2,
\end{equation}
which is better than the two bounds we obtained using Bernstein polynomial results. However, we remark that (\ref{eqn.paninskientropy}) is obtained using special properties of the entropy function and connections between KL-divergence and $\chi^2$-divergence \cite{Tsybakov2008}, which cannot be applied to general functions. Strukov and Timan~\cite{Strukov--Timan1977mathematical} also heavily exploited the structure of function $x^\alpha$ and $-x\ln x$ in order to analyze the Bernstein approximation error for these functions, and obtained tight-in-order results. 

\subsubsection{Bias of $H(\hat{P}_B)$}

We apply the general theory of positive linear operator approximation. The following lemma is a strengthened version of Lemma~\ref{lemma.paltaneathm251}. 
  \begin{lemma}\label{lemma_mod}
    If $F\colon C[0,1]\to \mathbb{R}$ is a linear positive functional and $F(e_0)=1$, then
    \begin{align}\label{eq:inequality2}
      |F(f)-f(x)| \le \omega^1(f,B_F(x);x) + \frac{5}{2}\omega^2_\varphi(f,h_2)
    \end{align}
    for all $f\in C[0,1]$ and $0<h_2\le \frac{1}{2}$, where $\varphi(x)=\sqrt{x(1-x)}$ and $h_2=\sqrt{V_F}/\varphi(x)$, and
    \begin{align}
      \omega^1(f,h;x) \triangleq \sup\left\{|f(u)-f(x)|:u\in[0,1],|u-x|\le h\right\}.
    \end{align}
    The ``bias'' $B_F(x)$ and ``variance'' $V_F(x)$ are defined in~(\ref{eqn.functionalbiasandvariance}). 
  \end{lemma}
  \begin{proof}
    Applying Lemma \ref{lemma.paltaneathm251} to $x=F(e_1)$ we have
    \begin{align}
      |F(f)-f(F(e_1))| \le \frac{5}{2}\omega^2_\varphi(f,h_2)
    \end{align}
    and then (\ref{eq:inequality2}) is the direct result of the triangle inequality $|F(f)-f(x)|\le |F(f)-f(F(e_1))| + |f(F(e_1))-f(x)|$.
  \end{proof}
  We show that Lemma \ref{lemma_mod} is indeed stronger than Lemma \ref{lemma.paltaneathm251}. Firstly, due to $h_1\ge h_2$, we have $\omega^2_\varphi(f,h_2)\le \omega^2_\varphi(f,h_1)$. Second, for $x\le 1/2$, we have
  \begin{align}
    \frac{B_F(x)}{2h_1\varphi(x)}\cdot \omega^1_\varphi(f,2h_1)
    &\approx \frac{B_F(x)}{2h_1\varphi(x)}\cdot \sup_{0\le s\le1}2h_1\varphi(s)f'(s)\\
    &\ge B_F(x)\cdot  \sup_{x\le s\le1-x}f'(s)\\
    &\approx \sup_{x\le s\le 1-x}\omega^1(f,B_F(x);s)
  \end{align}
  which is almost the supremum of $\omega^1(f,|F(e_1-xe_0)|;s)$ over $s\in[x,1-x]$ and is no less than the pointwise result $\omega^1(f,|F(e_1-xe_0)|;x)$, and here we have used the inequality $\varphi(s)\ge\varphi(x)$ for $x\le s\le 1-x$. A similar argument also holds for $x>1/2$. Hence, Lemma \ref{lemma_mod} transforms the first order term from the norm result in Lemma~\ref{lemma.paltaneathm251} to a pointwise result. 
  
Applying Lemma~\ref{lemma_mod} to the function $f(p) = -p\ln p$ and $F(f) = \mathbb{E} \left[ f \left( \frac{n\hat{p}+a }{n+Sa} \right) \right]$, where $n \cdot \hat{p} \sim \mathsf{B}(n,p)$, we have the following lemma. 
\begin{lemma}\label{lemma.hbbiasgeneral}
If $n\ge \max\{ Sa, 2ea, 4\}$, then
  \begin{align}
  & \sup_{P\in\mathcal{M}_S}| \bE_P H(\hat{P}_B) - H(P)| \nonumber \\
  & \quad \le  \frac{5nS\ln2}{(n+Sa)^2} + \frac{2Sa}{n+Sa}\ln\left(\frac{n+Sa}{2a}\right).
\end{align}
\end{lemma}  

Note that Lemma~\ref{lemma.hbbiasgeneral} implies a slightly weaker bias bound than Theorem~\ref{thm.upperbound}, but it is only sub-optimal up to a multiplicative constant. The bias bound in Theorem~\ref{thm.upperbound} is obtained using the following lemma, whose proof only applies to the entropy function. 
\begin{lemma}\label{lemma.hbbiasentropy}
If $n\ge \max\{2ea,Sa\}$,
\begin{align}
 &  \sup_{P\in\mathcal{M}_S}| \bE_PH(\hat{P}_B) - H(P)| \nonumber \\
 & \quad \le \ln\left(1+\frac{S-1}{n+Sa}\right) \vee \frac{2Sa}{n+Sa} \ln \left( \frac{n+Sa}{2a}\right).
\end{align}
\end{lemma}

\subsection{Bounding the variance}

The next lemma follows from an application of bounded difference inequality presented in Lemma~\ref{lemma.es}. 
\begin{lemma}\label{cor.vargeneralf}
The variance of $F(P_n)$ satisfies the following upper bound:
\begin{equation}
\mathsf{Var}(F(P_n)) \leq n\cdot\max_{0\leq j <n}(f((j+1)/n)-f(j/n))^2.
\end{equation}
If $f$ is monotone, then we can strengthen the bound to be
\begin{equation}
\mathsf{Var}(F(P_n)) \leq \frac{n}{4}\cdot\max_{0\leq j <n}(f((j+1)/n)-f(j/n))^2.
\end{equation}
\end{lemma}

We first bound the variance for $F_\alpha(P_n),\alpha>1$. We have
\begin{align}
\max_{0\leq j <n} (((j+1)/n)^\alpha - (j/n)^\alpha)^2 & \leq \left( 1- \left( 1- \frac{1}{n} \right)^\alpha \right)^2 \\ & \leq\left( \frac{\alpha}{n} \right)^2,
\end{align}
where in the last step we used Bernoulli's inequality: $(1+x)^r \geq 1+rx,\forall r \geq 1, x>-1, x\in \mathbb{R}$. Using Lemma~\ref{lemma.es}, we know the variance is upper bounded by
\begin{equation}
\mathsf{Var}(F_\alpha(P_n)) \leq \frac{\alpha^2}{4n}.
\end{equation}

We bound the variance of $F_\alpha(P_n), 0<\alpha<1$ in the following lemma.

\begin{lemma}\label{lemma.varianceboundfalphabetween01}
For $0<\alpha<1/2$, we have
\begin{align}
&    \sup_{P\in\mathcal{M}_S}\mathsf{Var}(F_\alpha(P_n)) \nonumber \\ 
& \quad \leq  \frac{10S}{n^{2\alpha}} \nonumber \\
& \qquad + \left(\frac{3\alpha\cdot2^{3+2\alpha}+1}{8\alpha^2}\left(\frac{8\alpha}{e}\right)^{2\alpha}+4\right)\left(\frac{S}{n^{2\alpha}}\wedge \frac{1}{n^{2\alpha-1}}\right) \\
& \quad \lesssim  \frac{S}{n^{2\alpha}} .
\end{align}
For $1/2\leq \alpha <1$, we have
\begin{align}
&    \sup_{P\in\mathcal{M}_S}\mathsf{Var}(F_\alpha(P_n)) \nonumber \\  
& \quad \leq \frac{10S^{2-2\alpha}}{n} \nonumber \\
& \qquad + \left(\frac{3\alpha\cdot2^{3+2\alpha}+1}{8\alpha^2}\left(\frac{8\alpha}{e}\right)^{2\alpha}+4\right)\left(\frac{S}{n^{2\alpha}}\wedge \frac{1}{n^{2\alpha-1}}\right) \\
& \quad \lesssim \frac{S^{2-2\alpha}}{n} + \left(\frac{S}{n^{2\alpha}}\wedge \frac{1}{n^{2\alpha-1}}\right). 
\end{align}
\end{lemma}

Further, one can show that for all $\alpha\in (0,1)$,
\begin{equation}
\frac{3\alpha\cdot2^{3+2\alpha}+1}{8\alpha^2}\left(\frac{8\alpha}{e}\right)^{2\alpha}+4 \leq \frac{120}{\alpha^2},
\end{equation}
which is used in Theorem~\ref{thm.main}.

Regarding the variance of $H(P_n)$, we have
\begin{lemma}\label{lemma.varianceboundentropymle}
\begin{align}
\sup_{P \in \mathcal{M}_S} \mathsf{Var}(H(P_n)) & \leq \frac{(\ln n)^2}{n}\wedge \frac{2(\ln S + 3)^2}{n} \\
& \lesssim \frac{(\ln S)^2 \wedge (\ln n)^2}{n}.
\end{align}
\end{lemma}

The variance of $H(\hat{P}_B)$ is upper bounded by the following lemma. 
\begin{lemma}\label{lemma.hbvariance}
The variance of $H(\hat{P}_B)$ is upper bounded as follows:
\begin{align}
  \mathsf{Var}\left(H(\hat{P}_B)\right) \le \frac{2n}{(n+Sa)^2}\left[3+\ln\left(\frac{n+Sa}{a+1}\wedge S\right)\right]^2.
\end{align}
\end{lemma}

\section{Proofs of the lower bounds}\label{sec.lowerboundproof}

\subsection{Lower bounds for estimation of $F_\alpha(P)$ when $\alpha\geq 3/2$}

We apply the van Trees inequality as presented in Lemma~\ref{lemma.vantrees}. 

It suffices to consider the restricted case of $S = 2$ and prove the $n^{-1}$ lower bound. Thus, the model is equivalent to observing a Binomial random variable $X \sim \mathsf{B}(n,p)$, and one aims to estimate the functional $\psi_\alpha(p) = p^\alpha + (1-p)^\alpha$. We have
\begin{align}
\psi'_\alpha(p) = \alpha p^{\alpha -1} -\alpha (1-p)^{\alpha-1}. 
\end{align}
The Fisher information for parameter $p$ under the Binomial model is $\mathcal{I}(p) = \frac{n}{p(1-p)}$. Suppose we impose prior $\lambda(p)$ on parameter $p$. The van Trees inequality implies 
\begin{align}
& \sup_{P \in \mathcal{M}_S} \mathbb{E}_P \left( F_\alpha(P_n) - F_\alpha(P) \right)^2 \nonumber \\
& \quad \geq \inf_{\hat{F_\alpha}} \sup_{P \in \mathcal{M}_S} \mathbb{E}_P \left( \hat{F_\alpha} - F_\alpha(P) \right)^2 \\
&  \quad \geq \mathbb{E} \left( \mathbb{E}[F_\alpha(P)|X_1^S] - F_\alpha(P) \right)^2 \quad (\text{Bayes risk}) \\
& \quad \geq \frac{( \int \left[ \alpha p^{\alpha -1} - \alpha (1-p)^{\alpha-1}  \right]\lambda(p) \mathrm{d}p  )^2}{ \mathbb{E}_\lambda\left[ \frac{n}{p(1-p)} \right] + \mathcal{I}(\lambda)} \\
& \quad = \frac{( \int \left[ \alpha p^{\alpha -1} - \alpha (1-p)^{\alpha-1}  \right]\lambda(p) \mathrm{d}p  )^2}{ n \cdot \mathbb{E}_\lambda\left[ \frac{1}{p(1-p)} \right] + \mathcal{I}(\lambda)}
\end{align}
where the second inequality follows from the fact that the Bayes risk under any prior is upper bounded by the minimax risk~\cite{Wald1950statistical}.  

Taking $\lambda(p)$ to be the Dirichlet prior with parameter $(a,b)$, i.e., 
\begin{align}
\lambda(p) = \frac{1}{B(a,b)} p^{a-1} (1-p)^{b-1}, a>2,b>2,
\end{align}
we can explicitly evaluate the integrals above. Here $B(a,b)$ is the Beta function.

Taking $a = 4,b = 3$, we have
\begin{align}
& \sup_{P \in \mathcal{M}_S} \mathbb{E}_P \left( F_\alpha(P_n) - F_\alpha(P) \right)^2 \nonumber \\
& \quad \geq \frac{\left( 60 \alpha (B(\alpha+3,3) - B(\alpha+2,4)) \right)^2}{5n + 45}. 
\end{align}
Taking $C_\alpha = 72 \alpha^2\left(  B(\alpha+3,3) - B(\alpha+2,4) \right)^2$, we have
\begin{align}
\sup_{P \in \mathcal{M}_S} \mathbb{E}_P \left( F_\alpha(P_n) - F_\alpha(P) \right)^2 \geq \frac{C_\alpha}{n},\quad \text{for all }n\geq 1.
\end{align}

Note that $C_\alpha>0$ for all $\alpha\geq 3/2$.

\subsection{Lower bounds for estimation of $F_\alpha(P)$ when $1<\alpha<3/2$}

The following lemma was proved in \cite{Braess--Sauer2004}.
\begin{lemma}\label{lemma.qn1lower}
Let $k\geq 4$ be an even number. Suppose that the $k$-th derivative of $f$ satisfies $f^{(k)} \leq 0$ in $(0,1)$, $Q_{k-1}$ is the Taylor polynomial of order $k-1$ to $f$ at some $x_1$ in $(0,1)$. Then for $x\in [0,1]$,
\begin{equation}
f(x) - B_n[f](x) \geq Q_{k-1} - B_n[Q_{k-1}](x).
\end{equation}
\end{lemma}

Consider $f_\alpha(x) = -x^\alpha, 1<\alpha<2, x\in [0,1]$. Applying Lemma~\ref{lemma.qn1lower} to $f_\alpha$, taking $k = 6$, we have the following result.

\begin{lemma}\label{lemma.biasalphalarge}
Suppose $f_\alpha(x) = -x^\alpha, 1<\alpha<2$ on $[0,1]$. For all $x\in (0,1)$, we have
\begin{align}
& f_\alpha(x) -B_n[f_\alpha](x) \nonumber \\
\quad & \geq \frac{\alpha(\alpha-1)x^{\alpha-2}(1-x)}{2n} \Bigg( x + \frac{(2-\alpha)(3\alpha-1)x}{12n} \nonumber \\
& \qquad + \frac{(2-\alpha)(5-3\alpha)}{12n} \Bigg)  + \frac{R_1(x)}{n^3} + \frac{R_2(x)}{n^4},
\end{align}
where
\begin{align}
R_1(x) & = \frac{\alpha(\alpha-1)(\alpha-2)(\alpha-3) x^{\alpha-3} (1-x)}{24} \nonumber \\
& \qquad \times  \Bigg( 1+2(1-x)((5-2\alpha)x + \alpha-4) \Bigg ), \\
R_2(x) & = \frac{\alpha(\alpha-1)(\alpha-2)(\alpha-3)(\alpha-4) }{120} \nonumber \\
& \qquad \times x^{\alpha-4} (1-x)(1-2x)(1-12x(1-x)). 
\end{align}
\end{lemma}

Note that we have assumed $S = cn,c>0$. If $c\leq 1$, we take a uniform distribution on $S$ elements $P = (1/S,1/S,\ldots,1/S)$, otherwise we take distribution $P = (n^{-1}-\epsilon,n^{-1} - \epsilon,\ldots,n^{-1} -\epsilon, \frac{n\epsilon}{S-n},\ldots,\frac{n\epsilon}{S-n})$, where $\epsilon$ will be taken to be arbitrarily small. We first analyze the $c\leq 1$ case. Applying Lemma~\ref{lemma.biasalphalarge}, we have
\begin{align}
& \sum_{i = 1}^S f_\alpha(1/S) - B_n[f_\alpha](1/S)\quad (\text{Note that } f_\alpha(x) = -x^\alpha) \nonumber \\
& = \bE F_\alpha(P_n) - F_\alpha(P)   \nonumber \\
& \geq S \cdot \Bigg( \frac{\alpha(\alpha-1)}{2 S^{\alpha-2}n} \left( \frac{1}{S} + \frac{(2-\alpha)(5-3\alpha)}{12n} \right) \nonumber \\
&  \quad + \frac{\alpha(\alpha-1)(\alpha-2)(\alpha-3)}{24S^{\alpha-3}n^3}(1+2(\alpha-4)) \nonumber \\
& \quad  + \frac{\alpha(\alpha-1)(\alpha-2)(\alpha-3)(\alpha-4)}{120S^{\alpha-4} n^4} + o(n^{-\alpha})\Bigg ) \nonumber \\
& = \frac{\alpha(\alpha-1)}{n^{\alpha-1}} \Bigg( \frac{1}{2c^{\alpha-3}}\left(\frac{1}{c}+\frac{(2-\alpha)(5-3\alpha)}{12}\right) \nonumber \\
& \quad + \frac{(\alpha-2)(\alpha-3)(1+2(\alpha-4))}{24 c^{\alpha-4}} \nonumber\\
& \quad + \frac{(\alpha-2)(\alpha-3)(\alpha-4)}{120c^{\alpha-5}} \Bigg) + o(n^{-(\alpha-1)}) \nonumber\\
& = \frac{\alpha(\alpha-1)c^{2-\alpha}}{n^{\alpha-1}} \Bigg(\frac{1}{2}+ \frac{(2-\alpha)(5-3\alpha)c}{24} \nonumber \\
& \quad +  \frac{(\alpha-2)(\alpha-3)(1+2(\alpha-4))c^{2}}{24 } \nonumber\\
& \quad + \frac{(\alpha-2)(\alpha-3)(\alpha-4)c^{3}}{120} \Bigg) + o(n^{-(\alpha-1)}) \nonumber \\
& \geq \frac{\alpha c^{2-\alpha}(124-330\alpha + 285 \alpha^2-90 \alpha^3 + 11 \alpha^4)}{120 n^{\alpha-1}} + o(n^{-(\alpha-1)}), \nonumber
\end{align}
where the first inequality follows from Lemma~\ref{lemma.biasalphalarge}, and in the last step we have taken $c = 1$ in the following expression
\begin{align}
& \frac{1}{2}+ \frac{(2-\alpha)(5-3\alpha)c}{24} +  \frac{(\alpha-2)(\alpha-3)(1+2(\alpha-4))c^{2}}{24 } \nonumber \\
& \quad + \frac{(\alpha-2)(\alpha-3)(\alpha-4)c^{3}}{120},
\end{align}
and considered the fact that it is a monotonically decreasing function with respect to $c$ on $(0,1]$ for any $\alpha\in (1,3/2)$.

For cases when $c>1$, since we take $P = (n^{-1}-\epsilon,n^{-1} - \epsilon,\ldots,n^{-1} -\epsilon, \frac{n\epsilon}{S-n},\ldots,\frac{n\epsilon}{S-n})$, by a continuity argument, the analysis is exactly the same as that above when we set $c=1$ as we can take $\epsilon$ as small as possible. One can verify that the function $\alpha(124-330\alpha + 285 \alpha^2-90 \alpha^3 + 11 \alpha^4)/120$ is positive on interval $(1,3/2)$. Defining $\sqrt{c_\alpha} = \alpha c^{2-\alpha}(124-330\alpha + 285 \alpha^2-90 \alpha^3 + 11 \alpha^4)/120>0$ when $c\leq 1$, and $\sqrt{c_\alpha} = \alpha (124-330\alpha + 285 \alpha^2-90 \alpha^3 + 11 \alpha^4)/120>0$ when $c>1$, the proof is completed.

\subsection{Lower bounds for estimation of $F_\alpha(P)$ when $0<\alpha<1$}

Applying Lemma~\ref{lemma.qn1lower} to function $f_\alpha(x) = x^\alpha, \alpha \in (0,1)$, taking $k = 4$, we have the following result:
\begin{lemma}\label{lemma.alphalower}
For $f_\alpha(x) = x^\alpha$ on $[0,1]$, $\alpha\in (0,1),  x\in (0,1)$, we have
\begin{equation}
f_\alpha(x) - B_n[f_\alpha](x) \geq  \frac{\alpha(1-\alpha)}{2n} x^{\alpha-2} (1-x) \left( x - \frac{2-\alpha}{3n} \right). 
\end{equation}
\end{lemma}

Suppose $n\geq S$. Define distribution $W = (w_1,w_2,\ldots,w_S) \in \mathcal{M}_S$ such that
\begin{equation}
1\leq i \leq S-1, w_i = \frac{1}{n}; \quad w_S = 1- \frac{S-1}{n}. 
\end{equation}
Note that $w_i \geq n^{-1}, 1\leq i\leq S$. It follows from Lemma~\ref{lemma.alphalower} that
\begin{align}
F_\alpha(W) - \bE_W F_\alpha(P_n) &  \geq \sum_{i = 1}^{S-1} \frac{\alpha(1-\alpha)}{6n^2} \left(\frac{1}{n} \right)^{\alpha-2} \left( 1-\frac{1}{n} \right) \\
& = \frac{\alpha(1-\alpha)(S-1)}{6n^\alpha}\frac{n-1}{n}.
\end{align}

Thus, we know for all $0<\alpha<1$,
\begin{align}
& \sup_{P \in \mathcal{M}_S} \bE_P \left( F_\alpha(P) - F_\alpha(P_n) \right)^2 \nonumber \\
& \quad \geq \frac{\alpha^2(1-\alpha)^2 (S-1)^2}{36 n^{2\alpha}} \left( 1-\frac{1}{n}\right)^2.
\end{align}

It is shown in \cite{Jiao--Venkat--Han--Weissman2015minimax} that the following minimax lower bound holds for estimation of $F_\alpha(P), 1/2\leq \alpha<1$.
\begin{lemma}\label{lemma.minimaxlowerfalpha12to1}
 For $\frac{1}{2}\le\alpha<1$, we have
  \begin{align}
  &  \inf_{\hat{F}}\sup_{P\in\mathcal{M}_S} \bE_P\left(\hat{F}-F_\alpha(P)\right)^2 \nonumber \\
  & \quad\ge \frac{\alpha^2}{32en}\Bigg [(2(S-1))^{1-\alpha}-2^{-\alpha} \nonumber \\
  & \qquad -\frac{1-\alpha}{4n}\left((2(S-1))^{1-\alpha}+2^{-\alpha}\right)\Bigg]^2 \nonumber \\
    & \qquad - e^{-n/4}S^{2(1-\alpha)}\nonumber \\
    & \quad \gtrsim \frac{S^{2-2\alpha}}{n},
  \end{align}
  where the infimum is taken over all possible estimators.
  \end{lemma}
Since this lower bound holds for all possible estimators, it also holds for the MLE $F_\alpha(P_n)$. Since $\max\{a,b\}\geq \frac{1}{2}(a+b)$, we have the desired lower bound. 

\subsection{Lower bounds for estimation of $H(P)$}

Braess and Sauer \cite{Braess--Sauer2004} derived the following lower bound for the approximation error of Bernstein polynomials for the function $g(x) = -x \ln x$:
\begin{lemma}\label{lemma.glower}
Define $g(x) = -x\ln x$ on $[0,1]$. For $x \geq \frac{15}{n}, x\in [0,1]$, we have
\begin{equation}
g(x) - B_n[g](x) \geq \frac{1-x}{2n} + \frac{1}{20n^2 x} - \frac{x}{12 n^2}.
\end{equation}
\end{lemma}

Applying Lemma~\ref{lemma.glower} to the estimation of $H(P)$, we know that if $\forall 1\leq i\leq S, p_i \geq \frac{15}{n}$,
\begin{equation}
H(P) - \bE H(P_n) \geq \frac{S-1}{2n} + \frac{1}{20n^2} \left( \sum_{i = 1}^S \frac{1}{p_i} \right) - \frac{1}{12n^2}.
\end{equation}

Consider the uniform distribution $P$ with $n\geq 15S$, which guarantees $p_i \geq \frac{15}{n}$. Since
\begin{equation}
\sum_{i = 1}^S \frac{1}{p_i} \geq S^2,
\end{equation}
we have
\begin{equation}
\sup_{P\in \mathcal{M}_S} \left(  H(P) - \bE H(P_n) \right)  \geq \frac{S-1}{2n} + \frac{S^2}{20n^2} - \frac{1}{12n^2}.
\end{equation}
Thus, when $n\geq 15S$,
\begin{equation}
\sup_{P\in \mathcal{M}_S} \bE_P \left( H(P) - H(P_n) \right)^2 \geq \left( \frac{S-1}{2n} + \frac{S^2}{20n^2} - \frac{1}{12n^2} \right)^2.
\end{equation}

It was shown in \cite[Prop. 1]{Wu--Yang2014minimax} that the following minimax lower bound holds.
\begin{lemma}\label{lemma.minimaxlowerentropyvariance}
There exists a universal constant $c>0$ such that
\begin{equation}
\inf_{\hat{H}} \sup_{P \in \mathcal{M}_S} \bE_P \left( \hat{H} - H(P) \right)^2 \geq c \frac{\ln^2 S}{n},
\end{equation}
where the infimum is taken over all possible estimators $\hat{H}$.
\end{lemma}

Hence, we have
\begin{align}
& \sup_{P\in \mathcal{M}_S} \bE_P \left( H(P) - H(P_n) \right)^2 \nonumber \\
& \quad \geq \max\left\{ \left( \frac{S-1}{2n} + \frac{S^2}{20n^2} - \frac{1}{12n^2} \right)^2 , c \frac{\ln^2 S}{n} \right\}  \\
& \quad \geq \frac{1}{2}\left( \frac{S-1}{2n} + \frac{S^2}{20n^2} - \frac{1}{12n^2} \right)^2 + \frac{c}{2} \frac{\ln^2 S}{n}.
\end{align}

Similar arguments can be applied to the Miller--Madow estimator.

\subsection{Lower bounds for entropy estimation using $H(\hat{P}_B)$}

Since $H(\hat{P}_B)$ is a specific estimator for entropy, the following lemma is proved via considering several specific distributions. 

\begin{lemma}\label{lemma.hbbiaslowerbound}
If $n\ge\max\{15S,Sa, 2ea\}$,
  \begin{align}
    & \sup_{P\in\mathcal{M}_S}\left |\bE_P H(\hat{P}_B) - H(P) \right| \nonumber \\
    & \quad \ge \frac{(S-1)a}{4(n+Sa)}\ln \left( \frac{n+Sa}{a}\right) + \frac{S-1}{8n} + \frac{S^2}{80n^2} - \frac{1}{48n^2}
  \end{align}
  If $n<Sa$, then
  \begin{equation}
    \sup_{P\in\mathcal{M}_S}\left |\bE_P H(\hat{P}_B) - H(P) \right| \geq  \frac{S-1}{2S} \ln S.
  \end{equation}
    If $n<2ea$, then
  \begin{equation}
    \sup_{P\in\mathcal{M}_S}\left |\bE_P H(\hat{P}_B) - H(P) \right| \geq \frac{S-1}{2e+S} \ln S.
  \end{equation}
  
  If $n<15S,n\geq 2ea$, then
  \begin{align}
  &  \sup_{P\in\mathcal{M}_S}\left |\bE_P H(\hat{P}_B) - H(P) \right|  \nonumber \\
  & \quad \geq \frac{(S-1)a}{4(n+Sa)}\ln \left( \frac{n+Sa}{a}\right) + \frac{\lfloor n/15 \rfloor}{8n} - \frac{1}{16n}. 
  \end{align}
\end{lemma}

The corresponding results in Theorem~\ref{thm.lowerbound} follow from Lemma~\ref{lemma.hbbiaslowerbound}, Lemma~\ref{lemma.minimaxlowerentropyvariance}, and the inequality $\max\{a,b\}\ge \frac{a+b}{2}$.

\subsection{Lower bounds for entropy estimation using $\hat{H}^{\mathsf{Bayes}}$}

We prove Theorem~\ref{thm.bayeslowerbound} below. Applying Lemma~\ref{lemma.digammafunctionbound}, we have
\begin{align}
\hat{H}^{\mathsf{Bayes}} & \leq \psi(Sa + n + 1) - \sum_{i = 1}^S \frac{a+X_i}{Sa + n}\psi(a+1) \\
& = \psi(Sa + n+1) - \psi(a+1) \\
& \leq \ln \left( \frac{Sa + n + e^{-\gamma}}{a + \frac{1}{2}} \right). 
\end{align}

Since $\hat{H}^{\mathsf{Bayes}}$ is upper bounded by $\ln \left( \frac{Sa + n + e^{-\gamma}}{a + \frac{1}{2}} \right)$ for any empirical observations, the squared error it incurs in Shannon entropy estimation when the true distribution is the uniform distribution is at least
\begin{align}
\left(  \ln \left( \frac{Sa + S/2}{Sa + n + e^{-\gamma}} \right) \right)^2
\end{align}
if $S\geq 2(n +1)$. 

\section*{Acknowledgments}

We thank Dany Leviatan, Gancho Tachev, and Radu Paltanea for very helpful discussions regarding the literature on approximation theory using positive linear operators. We thank Jayadev Acharya, Alon Orlitsky, Ananda Theertha Suresh, and Himanshu Tyagi for communicating to us the independent discovery that it suffices to take $n\gg 1$ samples to consistently estimate $F_\alpha(P)$, when $\alpha>1$. We thank Maya Gupta for raising the question on the optimality of the Dirichlet prior smoothing techniques applying to entropy estimation. We thank the anonymous reviewers and the associate editor for very helpful comments that significantly improved the presentation of the paper.  

\appendices

\section{Auxiliary lemmas}
We begin with the definition of the negative association property, which allows us to upper bound the variance by treating each component of the empirical distribution $P_n(i)$ as ``independent'' random variables.
\begin{definition}
  \cite[Def. 2.1]{joag-dev1983} Random variables $X_1,X_2,\cdots,X_S$ are said to be negatively associated if for any pair of disjoint subsets $A_1,A_2$ of $\{1,2,\cdots,S\}$, and any component-wise increasing functions $f_1,f_2$,
  \begin{align}
   \mathsf{Cov}\left(f_1(X_i,i\in A_1), f_2(X_j,j\in A_2)\right) \le 0.
  \end{align}
\end{definition}

To verify whether random variables $X_1,X_2,\cdots,X_S$ are negatively associated or not, the following lemma presents a useful criterion.
\begin{lemma}\label{lem_NA}
  \cite[Thm. 2.9]{joag-dev1983} Let $X_1,X_2,\cdots,X_S$ be $S$ independent random variables with log-concave densities. Then the joint conditional distribution of $X_1,X_2,\cdots,X_S$ given $\sum_{i=1}^S X_i$ is negatively associated.
\end{lemma}

In light of the preceding lemma, we can obtain the following corollary.
\begin{corollary}\label{cor_NA_Multinomial}
  For any discrete probability distribution vector $P\in\mathcal{M}_S$, the random variables $\mathbf{X}=(X_1,X_2,\cdots,X_S)$ drawn from the multinomial distribution $\mathbf{X}\sim\mathsf{multi}(n;P)$ are negatively associated.
\end{corollary}
\begin{proof}
  Consider the Poissonized model $Y_i\sim \mathsf{Poi}(np_i), 1\le i\le S$ with all $Y_i$ independent, it is straightforward to verify that each $Y_i$ possesses a log-concave distribution. Then conditioning on $\sum_{i=1}^S Y_i=n$, we know that $(Y_1,Y_2,\cdots,Y_S)|(\sum_{i=1}^S Y_i=n) \sim \mathsf{multi}(n;P)$, hence Lemma \ref{lem_NA} yields the desired result.
\end{proof}

The next lemma gives bounds on the digamma functions $\psi(z) = \frac{\Gamma'(z)}{\Gamma(z)}$. 

\begin{lemma} \label{lemma.digammafunctionbound} \cite[Lemma 1.7]{batir2008inequalities}
The digamma function $\psi(z)$ is the only solution of the functional equation $F(x+1) = F(x) + \frac{1}{x}$ that is monotone, strictly concave on $\mathbb{R}_+$ and satisfies $F(1) = -\gamma$, where $\gamma \approx 0.5772$ is the Euler--Mascheroni constant.

Let $x$ be a positive real number. Then, 
\begin{align}
\ln(x + 1/2) < \psi(x+1) \leq \ln(x+ e^{-\gamma}).
\end{align}
If $x\geq 1$, then
\begin{align}
\ln(x+1/2) < \psi(x+1) \leq \ln(x + e^{1-\gamma} -1). 
\end{align}
\end{lemma}

The following lemma gives some tail bounds for Poisson or Binomial random variables.
\begin{lemma}\label{lem_chernoff}
\cite[Exercise 4.7]{mitzenmacher2005probability} If $X\sim \mathsf{Poi}(\lambda)$ or $X\sim \mathsf{B}(n,\frac{\lambda}{n})$, then for any $\delta>0$, we have
\begin{align}
\bP(X \geq (1+\delta) \lambda) & \leq \left( \frac{e^\delta}{(1+\delta)^{1+\delta}} \right)^\lambda, \\
\bP(X \leq (1-\delta) \lambda) & \leq  \left( \frac{e^{-\delta}}{(1-\delta)^{1-\delta}} \right)^\lambda\leq  e^{-\delta^2\lambda/2}.
\end{align}
\end{lemma}

To establish the upper bound of the variance obtained by the plug-in estimator $F_\alpha(P_n)$, we split $p$ into two different regimes $p\le 1/n$ or $p>1/n$, and the following lemmas give the corresponding variance bounds.
\begin{lemma}\label{lem_variance_p_small}
  For $nX\sim\mathsf{B}(n,p),p\le 1/n$, we have
  \begin{align}
    \mathsf{Var}(X^\alpha) \le  \frac{2}{n^{2\alpha}}\wedge \frac{2p}{n^{2\alpha-1}}  \quad 0<\alpha<1 .
  \end{align}
\end{lemma}
\begin{lemma}\label{lem_variance_p_big}
  For $nX\sim\mathsf{B}(n,p),p \geq  1/n,0<\alpha<1$, we have
  \begin{align}
\mathsf{Var}(X^\alpha) & \le  \frac{10p^{2\alpha-1}}{n} + \frac{3}{2\alpha}\left(\frac{16\alpha}{en}\right)^{2\alpha} + \frac{2}{n^{2\alpha}} + \frac{1}{8\alpha^2}\left(\frac{8\alpha}{en}\right)^{2\alpha}.  
\end{align}
\end{lemma}

\section{Proofs of main lemmas}
\subsection{Proof of Lemma~\ref{lemma.biasgeneralf}}
We compute the first moment of $F(P_n)$.
\begin{equation}
\bE F(P_n) = \sum_{j = 0}^n f \left( \frac{j}{n} \right) \bE h_j,
\end{equation}
and
\begin{align}
\bE h_j & = \bE \sum_{i = 1}^S \mathbbm{1}(X_i = j) \\
&  = \sum_{i = 1}^S \bP(X_i = j) \\
& = \sum_{i = 1}^S \binom{n}{j} p_i^j(1-p_i)^{n-j}.
\end{align}
Thus, we have
\begin{align}
\bE F(P_n) & = \sum_{j = 0}^n f \left( \frac{j}{n} \right)\sum_{i = 1}^S \binom{n}{j} p_i^j(1-p_i)^{n-j} \\
& = \sum_{j = 0}^n \sum_{i = 1}^S f \left( \frac{j}{n} \right) \binom{n}{j} p_i^j(1-p_i)^{n-j}.
\end{align}

The bias of $F(P_n)$ is
\begin{align}
& \mathsf{Bias}(F(P_n)) \nonumber \\
& \quad = \bE F(P_n) - F(P)\\
& \quad =  \sum_{i = 1}^S \left( \sum_{j = 0}^n f \left( \frac{j}{n} \right) \binom{n}{j} p_i^j(1-p_i)^{n-j} - f(p_i) \right).
\end{align}

\subsection{Proof of Lemma~\ref{lemma.dtmoduluscomputation}}

We first compute the second-order modulus. Fix $t, 0< t\leq 1/2$. Defining $M \triangleq \frac{u+v}{2}$, then the computation of second-order modulus is equivalent to maximization of $|f(M-t) - 2 f(M) + f(M+t)|$ over interval $M \in [t, 1-t]$, since all the functions we consider are strictly convex or concave over $[0,1]$.

For $f(x) = x^\alpha,0<\alpha<1$, $f(x)$ is strictly concave on $[0,1]$. It follows from Jensen's inequality that
\begin{equation}
g(M) = (M-t)^\alpha - 2M^\alpha + (M+t)^\alpha \leq 0,
\end{equation}
and it suffices to minimize this function of $M$ in order to obtain the modulus. Taking derivative of $g(M)$, we have
\begin{equation}
g'(M) = \alpha \left( (M-t)^{\alpha-1} - 2 M^{\alpha-1} + (M+t)^{\alpha-1} \right) \geq 0,
\end{equation}
since $x^{\alpha-1}$ is a convex function on $[t,1-t]$. It implies that the function $g(M)$ is non-decreasing, and the minimum of $g(M)$ over $M\in [t,1-t]$ is attained at $M = t$, and the minimum value is $g(t) = (2^\alpha-2) t^\alpha$. Hence, the corresponding second-order modulus is $|2-2^\alpha| t^\alpha$.

Analogous procedures computes the second-order modulus for $x^\alpha,1<\alpha<2$ and $-x\ln x$.

Now we consider the computation of Ditzian--Totik second-order modulus. Fix $t, 0<t\leq 1$. Again denote $M \triangleq \frac{u+v}{2} \in [0,1]$. Then the optimization is over the regime $|u-v| \leq 2t\varphi(M) = 2t \sqrt{M(1-M)}$. Equivalently, it is the interval $[M-t\sqrt{M(1-M)},M+t\sqrt{M(1-M)}] \cap [0,1]$.

Since the function $f(x) = -x\ln x$ is strictly convex on $[0,1]$, the maximum of $\left|f(u)-2f\left(\frac{u+v}{2}\right)+f(v)\right|$ is definitely attained when $u$ and $v$ take the boundary values of the feasible interval $[M-t\sqrt{M(1-M)},M+t\sqrt{M(1-M)}] \cap [0,1]$.

Define $\Delta \triangleq t\sqrt{\frac{1-M}{M}}$. The feasible interval can be equivalently written as $[M-\Delta M, M+\Delta M] \cap [0,1]$. We have
\begin{equation}
M - t\sqrt{M(1-M)} \geq 0 \Leftrightarrow M \geq \frac{t^2}{1+t^2},
\end{equation}
as well as
\begin{equation}
M + t\sqrt{M(1-M)} \leq 1 \Leftrightarrow M \leq \frac{1}{1+t^2}.
\end{equation}

Hence, it is equivalent to maximize over three regimes:
\begin{enumerate}
\item Regime A:

$u = 0, v = 2M, 0\leq M \leq \frac{t^2}{1+t^2}$.
\item Regime B:

$u = M-\Delta M, v = M + \Delta M, M \in \left[ \frac{t^2}{1+t^2}, \frac{1}{1+t^2} \right]$
\item Regime C:

$u = 2M-1, v = 1, 1 \geq M \geq \frac{1}{1+t^2}$.
\end{enumerate}

Over regime A, we have
\begin{equation}
\left|f(u)-2f\left(\frac{u+v}{2}\right)+f(v)\right| = 2M \ln 2.
\end{equation}
Maximizing over $0\leq M \leq \frac{t^2}{1+t^2}$, the maximum value is $\frac{t^2 \ln 4}{1+t^2}$, attained at $M = \frac{t^2}{1+t^2}$.

Over regime C, we have
\begin{equation}
\left|f(u)-2f\left(\frac{u+v}{2}\right)+f(v)\right| = \left | 2M \ln M - (2M-1) \ln (2M-1) \right|.
\end{equation}
Maximizing over $\frac{1}{1+t^2} \leq M \leq 1$, the maximum is attained at $M = \frac{1}{1+t^2}$, and the maximum value is no more than $\frac{t^2 \ln 4}{1+t^2}$.

Now we consider regime B. Since $M  \in \left[ \frac{t^2}{1+t^2}, \frac{1}{1+t^2} \right]$ in regime B, we know $\Delta = t\sqrt{\frac{1-M}{M}} \in [t^2,1]$. We have

\begin{align}
& \left|f(u)-2f\left(\frac{u+v}{2}\right)+f(v)\right| \nonumber \\
& \quad = M\left | (1-\Delta)\ln(1-\Delta) + (1+\Delta)\ln(1+\Delta) \right|.
\end{align}
Since $\Delta = t\sqrt{\frac{1-M}{M}}$ implies $M = \frac{t^2}{t^2 + \Delta^2}$, we can recast the corresponding optimization problem as maximizing
\begin{equation}
\frac{t^2}{\Delta^2 + t^2} \left | (1-\Delta)\ln(1-\Delta) + (1+\Delta)\ln(1+\Delta) \right|
\end{equation}
subject to constraint $\Delta \in [t^2, 1]$. One can show that the maximum is always attained at $\Delta = 1$, with the maximum value $\frac{t^2 \ln 4}{1+t^2}$.

To sum up, we conclude that when $0< t\leq 1$, the maximum of the optimization problem defining $\omega^2_\varphi(-x\ln x,t)$ is always attained at $u = 0, v = \frac{2t^2}{1+t^2}$, with the resulting modulus $\frac{t^2 \ln 4}{1+t^2}$.

Analogous computation can also be done for function $x^\alpha,0<\alpha<1$. For the function $x^\alpha,1<\alpha<2$, it is hard to compute the modulus exactly, but it is easy to show that it is of order $t^2$.

\subsection{Proof of Lemma~\ref{lemma.biasalphaonetwo}}
The bias of $F_\alpha(P_n), 1<\alpha<2$ can be expressed as follows:
\begin{align}
| \mathsf{Bias}(F_\alpha(P_n))| & = | \bE \sum_{i = 1}^S P_n^\alpha(i) - p_i^\alpha | \\
& \leq \left| \bE \sum_{i: p_i \leq \frac{1}{n}} P_n^\alpha(i) - p_i^\alpha \right| \nonumber \\
& \quad  + \left| \bE \sum_{i: p_i > \frac{1}{n}} P_n^\alpha(i) - p_i^\alpha \right| \\
& \triangleq B_1 + B_2.
\end{align}

Now we bound $B_1$ and $B_2$ separately. It follows from Jensen's inequality that for any $i$,
\begin{equation}
\bE P_n^\alpha(i) \geq p_i^\alpha, \quad 1\leq i \leq S.
\end{equation}

Hence, we have
\begin{align}
B_1 & = \bE \sum_{i: p_i \leq \frac{1}{n}} P_n^\alpha(i) - p_i^\alpha \\
& = \bE \sum_{i: p_i \leq \frac{1}{n}} P_n^\alpha(i) -  \frac{(np_i)^\alpha}{n^\alpha}  \\
& \leq \bE \sum_{i: p_i \leq \frac{1}{n}} P_n^\alpha(i) -  \frac{(np_i)^2}{n^\alpha}  \\
& = \bE \sum_{i: p_i \leq \frac{1}{n}} \frac{(n P_n(i))^\alpha}{n^\alpha}-  \frac{(np_i)^2}{n^\alpha}  \\
& \leq \bE \sum_{i: p_i \leq \frac{1}{n}} \frac{(n P_n(i))^2}{n^\alpha} -  \frac{(np_i)^2}{n^\alpha} \\
& = \sum_{i: p_i \leq \frac{1}{n}}  \frac{\bE(n P_n(i))^2}{n^\alpha}-  \frac{(np_i)^2}{n^\alpha}  \\
& = \sum_{i: p_i \leq \frac{1}{n}} \frac{(n p_i)^2 + np_i(1-p_i)}{n^\alpha}-  \frac{(np_i)^2}{n^\alpha} \\
& \leq \sum_{i: p_i \leq \frac{1}{n}} \frac{np_i}{n^\alpha}\\
& = \frac{1}{n^{\alpha-1}},
\end{align}
where we have used the fact that $nP_n(i) \geq 1$ for any $P_n(i)\neq 0$.

Regarding $B_2$, we have the following bounds:
\begin{align}
B_2 & = \left| \bE \sum_{i: p_i >\frac{1}{n}} P_n^\alpha(i) - p_i^\alpha \right| \\
& \leq \sum_{i: p_i > \frac{1}{n}} \bE | P_n^\alpha(i) - p_i^\alpha | \\
& \leq \sum_{i: p_i > \frac{1}{n}}  \frac{3}{2}|2-2^\alpha| \left( \frac{p_i(1-p_i)}{n} \right)^{\alpha/2} \\
& \leq \frac{3|2-2^\alpha|}{2n^{\alpha/2}} \sum_{i: p_i > \frac{1}{n}} p_i^{\alpha/2},
\end{align}
where the second inequality follows from the pointwise estimate in Lemma~\ref{lemma.bernsteinerror}.

Denoting $| \{ i: p_i > \frac{1}{n} \} | = K \leq n$, we know
\begin{equation}
\sum_{i: p_i > \frac{1}{n}} p_i^{\alpha/2} \leq K^{1-\alpha/2} \leq  n^{1-\alpha/2},
\end{equation}
which implies that
\begin{equation}
B_2 \leq \frac{3|2-2^\alpha|}{2n^{\alpha/2}}n^{1-\alpha/2} = \frac{3|2-2^\alpha|}{2n^{\alpha-1}}.
\end{equation}

Therefore, we have
\begin{align}
| \mathsf{Bias}(F_\alpha(P_n))| & \leq  B_1 + B_2 \\
& \leq \frac{1}{n^{\alpha-1}}+ \frac{3|2-2^\alpha|}{2n^{\alpha-1}} \\
& \leq \frac{3|2-2^\alpha| + 2}{2n^{\alpha-1}} \\
& \leq \frac{4}{n^{\alpha-1}}.
\end{align}

\subsection{Proof of Lemma~\ref{lemma.hbbiasgeneral}}

We apply Lemma~\ref{lemma_mod}. Note that $h_2 = \frac{\sqrt{n}}{n+Sa}$. In order to ensure that $h_2 \leq 1/2$, it suffices to take $n\geq 4$. Also, since $n\geq Sa$, for any $i, 1\leq i\leq S$,
  \begin{align}
    \frac{|1-p_iS|a}{n+Sa} \le \frac{Sa}{n+Sa} \le \frac{1}{2}.
  \end{align}  

Meanwhile, since the function $\sum_{i=1}^S \frac{|1-p_i S|a}{n+Sa}$ is a convex function of $P = (p_1,p_2,\ldots,p_S)$, it attains its maximum at one of the corner points of the simplex. Hence,
\begin{align}\label{eq:convexcornerpoints}
\sum_{i=1}^S \frac{|1-p_i S|a}{n+Sa} & \leq \frac{|1-S|a}{n+Sa} + (S-1) \cdot \frac{a}{n+Sa} \\
& = \frac{2(S-1)a}{n+Sa}. 
\end{align}

  In light of Lemma \ref{lemma_mod}, we have
  \begin{align}
 &  | \bE_P H(\hat{P}_B) - H(P)| \nonumber \\
 & \quad\le \sum_{i=1}^S \left(\omega^1\left(f,\frac{|1-p_iS|a}{n+Sa};p_i\right) + \frac{5n\ln 2}{(n+Sa)^2}\right)\\
    & \quad \stackrel{(a)}{\le}  -\left(\sum_{i=1}^S\frac{|1-p_iS|a}{n+Sa}\right)\ln\left(\frac{1}{S}\sum_{i=1}^S\frac{|1-p_iS|a}{n+Sa}\right) \nonumber \\
    & \qquad \quad + \frac{5nS\ln2}{(n+Sa)^2}\\
    & \quad \stackrel{(b)}{\le} \frac{2Sa}{n+Sa}\ln\left(\frac{n+Sa}{2a}\right) + \frac{5nS\ln2}{(n+Sa)^2},
  \end{align}
where $(a)$ follows from the fact that if $|x-y|\leq 1/2, x,y\in [0,1]$, then $|x\ln x - y\ln y|\leq -|x-y|\ln |x-y|$~\cite[Thm. 17.3.3]{Cover--Thomas2006} and Jensen's inequality. Step~$(b)$ follows from the fact that the function $-y\ln y$ is monotonically increasing on the interval $[0,e^{-1}]$, and 
\begin{align}
\frac{1}{S} \sum_{i = 1}^S \frac{|1-p_i S|a }{n+Sa} & \leq \frac{2a}{n+Sa} \\
& \leq \frac{2a}{n} \\
& \leq e^{-1},
\end{align}
where in the last step we used the assumption that $n \geq 2ea$.

\subsection{Proof of Lemma~\ref{lemma.hbbiasentropy}}

We have
\begin{align}
H(\hat{P}_B) & = \sum_{i = 1}^S-\hat{p}_{B,i} \ln \hat{p}_{B,i} \\
& = H(P_B) + \sum_{i =1}^S (p_{B,i} - \hat{p}_{B,i}) \ln p_{B,i} - \sum_{i = 1}^S \hat{p}_{B,i} \ln \frac{\hat{p}_{B,i}}{p_{B,i}}.
\end{align}

Taking expectations on both sides, we have
\begin{equation}
\bE H(\hat{P}_B) - H(P) = H(P_B) - H(P) - \bE D(\hat{P}_B \| P_B),
\end{equation}
where $D(P\|Q) = \sum_{i = 1}^S p_i \ln \frac{p_i}{q_i}$ is the KL divergence between distributions $P$ and $Q$. Since $H(P_B) = H \left( \frac{n}{n+Sa} P + \frac{Sa}{n+Sa} U_S \right)$, where $U_S$ denotes the uniform distribution with alphabet size $S$, it follows from Jensen's inequality and the concavity of the entropy function that
\begin{align}
H(P_B) & \geq \frac{n}{n+Sa}H(P) + \frac{Sa}{n+Sa} H(U_S) \\
& \geq H(P).
\end{align}
Hence,
\begin{align}
\left| \bE H(\hat{P}_B) - H(P) \right| & \leq \max\{H(P_B) - H(P), \bE D(\hat{P}_B \| P_B)\}. 
\end{align}

In order to analyze the bias, it suffices to analyze the two terms separately. We first analyze $\bE D(\hat{P}_B \| P_B)$.

It follows from Jensen's inequality that
\begin{equation}
D(P \| Q) = \sum_{i = 1}^S p_i \ln \frac{p_i}{q_i} \leq \ln \left( \sum_{i = 1}^S \frac{p_i^2}{q_i} \right),
\end{equation}
whose derivation here follows from Tsybakov~\cite[Lemma 2.7]{Tsybakov2008}.

By Jensen's inequality, we have
\begin{equation}
\bE D(\hat{P}_B \| P_B) \leq \bE \ln \left( \sum_{i = 1}^S \frac{\hat{p}_{B,i}^2}{p_{B,i}} \right) \leq \ln \left( \sum_{i = 1}^S \frac{\bE \hat{p}_{B,i}^2 }{p_{B,i}} \right).
\end{equation}

We also have
\begin{equation}
\sum_{i = 1}^S \frac{p_i^2}{q_i} = 1 + \sum_{i =1}^S \frac{(p_i - q_i)^2}{q_i}, 
\end{equation}
and that
\begin{equation}
\bE (\hat{p}_{B,i} - p_{B,i})^2 = \frac{n^2}{(n+Sa)^2} \bE (\hat{p}_i - p_i)^2 = \frac{n p_i(1-p_i)}{(n+Sa)^2}.
\end{equation}

Hence,
\begin{align}
\bE D(\hat{P}_B \| P_B) & \leq \ln \left( 1+ \sum_{i = 1}^S \frac{n p_i(1-p_i)}{(n+Sa)(np_i+a)} \right) \\
& = \ln \left( 1+ \sum_{i = 1}^S \frac{n p_i(1-p_i)}{(n+Sa)np_i} \frac{np_i}{np_i+a} \right) \\
& \leq \ln \left( 1+ \sum_{i = 1}^S \frac{1-p_i}{(n+Sa)}  \right),
\end{align}
which implies that
\begin{equation}
\bE D(\hat{P}_B \| P_B)  \leq \ln \left( 1+\frac{S-1}{n+Sa} \right).
\end{equation}

Now we consider the deterministic gap $H(P_B) - H(P)$. It follows from a refinement result of Cover and Thomas~\cite[Thm. 17.3.3]{Cover--Thomas2006} that when $| p_{B,i} - p_i | \leq 1/2$ for all $i$, we have
\begin{align}
|H(P_B) - H(P)| & \leq - \| P_B - P \|_1 \ln \frac{\| P_B - P \|_1}{S} \\
& = S \cdot f\left( \frac{\| P_B - P \|_1}{S} \right),
\end{align}
where $f(x) = -x\ln x, x\in [0,1]$. Note that the condition $n\ge Sa$ ensures that $| p_{B,i} - p_i | \leq 1/2$.

We have
\begin{align}
\frac{1}{S} \| P_B - P \|_1 & = \frac{1}{S} \sum_{i = 1}^S \frac{Sa}{n+Sa} |p_i - 1/S| \\
& = \frac{1}{S} \sum_{i = 1}^S \frac{|1-p_i S|a}{n+Sa} \\
& \leq \frac{2a}{n+Sa},
\end{align}
where the last step follows from~(\ref{eq:convexcornerpoints}). Since we have assumed $n \geq 2ea$, we have $\frac{2a}{n+Sa} \leq \frac{2a}{n} \leq e^{-1}$. Since the function $f(x) = -x\ln x$ is monotonically increasing on the interval $[0,e^{-1}]$, we know
\begin{align}
|H(P_B) - H(P)| & \leq \frac{2Sa}{n+Sa} \ln \frac{n+Sa}{2a}. 
\end{align}

\subsection{Proof of Lemma~\ref{cor.vargeneralf}}

In our case, apparently $F(P_n)$ is a function of $n$ independent random variables $\{Z_i\}_{1\leq i\leq n}$ taking values in $\mathcal{Z} = \{1,2,\ldots,S\}$. Changing one location of the sample would make some symbol with count $j$ to have count $j+1$, and another symbol with count $i$ to have count $i-1$. Then the total change in the functional estimator is
\begin{equation}
f\left(\frac{j+1}{n}\right) - f \left( \frac{j}{n}\right) - f\left(\frac{i}{n}\right) + f\left(\frac{i-1}{n}\right).
\end{equation}

If $f$ is monotone, then the total change would be upper bounded by $\max_{0\leq j<n} |f((j+1)/n) - f(j/n)|$. If $f$ is not monotone, the total change can be upper bounded by $2 \cdot\max_{0\leq j<n} |f((j+1)/n) - f(j/n)|$. Applying Lemma~\ref{lemma.es}, we have the desired bounds.

\subsection{Proof of Lemma~\ref{lemma.varianceboundfalphabetween01}}

In light of Lemma \ref{lem_variance_p_small} and \ref{lem_variance_p_big}, we have
\begin{align}
& \sum_{i=1}^S \mathsf{Var}(P_n(i)^\alpha) \nonumber \\
  &= \sum_{i: p_i\le 1/n} \mathsf{Var}(P_n(i)^\alpha) + \sum_{i: p_i> 1/n} \mathsf{Var}(P_n(i)^\alpha)\\
  &\le \sum_{i: p_i\le 1/n} \frac{2}{n^{2\alpha}}\wedge \frac{2p_i}{n^{2\alpha-1}} \nonumber \\
  & \quad \quad + \sum_{i: p_i>1/n}\Bigg(\frac{10p_i^{2\alpha-1}}{n} + \frac{3}{2\alpha}\left(\frac{16\alpha}{en}\right)^{2\alpha} \nonumber \\
  & \qquad \qquad + \frac{2}{n^{2\alpha}} + \frac{1}{8\alpha^2}\left(\frac{8\alpha}{en}\right)^{2\alpha}\Bigg).
\end{align}
We obtain the desired bounds after using the concavity of $x^{2\alpha-1}$ when $1/2\le \alpha<1$. 

Now we exploit the negative association property of all random variables $P_n(i),1\le i\le S$.  Corollary \ref{cor_NA_Multinomial} and the monotonically increasing property of $x^\alpha$ yield
\begin{align}
  \mathsf{Var}(F_\alpha(P_n)) & = \sum_{i=1}^S \mathsf{Var}(P_n(i)^\alpha) \nonumber \\
  & \quad + 2\sum_{1\le i<j\le S}\mathsf{Cov}(P_n(i)^\alpha,P_n(j)^\alpha) \\
  & \le \sum_{i=1}^S \mathsf{Var}(P_n(i)^\alpha),
\end{align}
which finishes the proof of Lemma~\ref{lemma.varianceboundfalphabetween01}.

\subsection{Proof of Lemma~\ref{lemma.varianceboundentropymle}}

The upper bound $(\ln n)^2/n$ follows from Lemma~\ref{cor.vargeneralf}. We apply the Efron--Stein inequality (Lemma~\ref{lemma.esnew}) to obtain the other bound. Denote the $n$ i.i.d. samples from distribution $P$ as $Z_1,Z_2,\ldots,Z_n \in \mathcal{Z}$. Denoting the MLE $H(P_n)$ as $\hat{H}(Z_1,Z_2,\ldots,Z_n)$, since it is invariant to any permutation of $\{Z_1,Z_2,\ldots,Z_n\}$, we know that the Efron--Stein inequality implies
\begin{equation}\label{eqn.entropyesnew}
\mathsf{Var}(H(P_n)) \leq \frac{n}{2} \bE \left( \hat{H}(Z_1,Z_2,\ldots,Z_n) - \hat{H}(Z_1',Z_2,\ldots,Z_n) \right)^2,
\end{equation}
where $Z_1'$ is an i.i.d. copy of $Z_1$.

Recall that
\begin{equation}
X_i = \sum_{j = 1}^n \mathbbm{1}(Z_j = i),\quad 1\leq i \leq S.
\end{equation}

For notional brevity, we denote the $S$-tuple $(X_1,X_2,\ldots,X_S)$ as $X_1^S$, and the $n$-tuple $(Z_1,Z_2,\ldots,Z_n)$ as $Z_1^n$. A specific realization of $(X_1,X_2,\ldots,X_S)$ is denoted by $x_1^S = (x_1,x_2,\ldots,x_S)$, and a specific realization of $(Z_1,Z_2,\ldots,Z_n)$ is denoted by $z_1^n = (z_1,z_2,\ldots,z_n)$.

In order to upper bound the right hand side of (\ref{eqn.entropyesnew}), we first condition on $\{X_1,X_2,\ldots,X_S\}$. In other words, we use
\begin{align}
& \bE \left( \hat{H}(Z_1,Z_2,\ldots,Z_n) - \hat{H}(Z_1',Z_2,\ldots,Z_n) \right)^2 \\
 & \quad =\mathbb{E}\left[ \mathbb{E}\left[ \left( \hat{H}(Z_1,Z_2,\ldots,Z_n) - \hat{H}(Z_1',Z_2,\ldots,Z_n) \right)^2 \Bigg | X_1^S \right] \right]. 
\end{align}

The following lemma calculates the conditional distribution of $Z_1$ conditioned on $(X_1,X_2,\ldots,X_S)$.
\begin{lemma}\label{lemma.conditionalpaninski}
The conditional distribution of $Z_1$ conditioned on $(X_1,X_2,\ldots,X_S)$ is given by the following discrete distribution:
\begin{equation}
(X_1/n,X_2/n,\ldots,X_S/n).
\end{equation}
\end{lemma}

\begin{proof}
By definition of conditional distribution, for any $k, 1\leq k \leq S$, we have
\begin{align}
& \bP(Z_1=k| X_1^S = x_1^S) \nonumber \\
& \quad = \frac{\bP(Z_1 = k, X_1^S = x_1^S)}{\bP(X_1^S = x_1^S)} \\
& \quad = \frac{\bP(Z_1 = k)\bP(X_1^S = x_1^S|Z_1 = k)}{\bP(X_1^S = x_1^S)} \\
& \quad = \frac{p_k \binom{n-1}{x_1,x_2,\ldots,x_k-1,\ldots,x_S} p_k^{x_k-1} \prod_{i \neq k} p_i^{x_i}}{\binom{n}{x_1,x_2,\ldots,x_S} \prod_{1\leq i\leq S} p_i^{x_i}} \\
& \quad= \frac{\binom{n-1}{x_1,x_2,\ldots,x_k-1,\ldots,x_S}}{\binom{n}{x_1,x_2,\ldots,x_S}} \\
& \quad = \frac{x_k}{n},
\end{align}
where the multinomial coefficient $\binom{n}{x_1,x_2,\ldots,x_S}$ is defined as
\begin{equation}
\binom{n}{x_1,x_2,\ldots,x_S} = \frac{n!}{\prod_{i = 1}^S x_i!}.
\end{equation}
\end{proof}

Denoting $r(p) = -p\ln p$, we have $r(j/n) \triangleq\frac{-j}{n}\ln \frac{j}{n}$. We rewrite
\begin{equation}
\hat{H}(Z_1',Z_2,\ldots,Z_n)-\hat{H}(Z_1,Z_2,\ldots,Z_n)  = D_- + D_+,
\end{equation}
where
\begin{align}
D_- & = r\left( \frac{X_{Z_1}-1}{n} \right) - r\left(\frac{X_{Z_1}}{n}\right)\\
D_+ & = \begin{cases} r\left(\frac{X_{Z_1'+1}}{n} \right) - r\left(\frac{X_{Z_1'}}{n}\right) & Z_1 \neq Z_1'\\ r\left(\frac{X_{Z_1'}}{n} \right) - r\left(\frac{X_{Z_1'-1}}{n}\right) & Z_1 = Z_1' \end{cases}
\end{align}

Here, $D_-$ is the change in $\hat{H}$ that occurs when $Z_1$ is removed according to the distribution specified in Lemma~\ref{lemma.conditionalpaninski}, and $D_+$ is the change in $\hat{H}$ that occurs when $Z_1'$ is added back according to the true distribution $P$.

Now we compute $\bE [D_-^2|X_1^S]$ and $\bE [D_+^2| X_1^S]$. We have
\begin{align}
\bE[D_-^2|X_1^S] & = \sum_{1\leq i\leq S} \frac{X_i}{n} \left( r\left(\frac{X_i-1}{n}\right) - r\left(\frac{X_i}{n}\right) \right)^2,
\end{align}
and
\begin{align}
& \bE[D_+^2|X_1^S] \nonumber \\
 & \quad = \sum_{1\leq i\leq S} p_i  \frac{X_i}{n} \left(  r\left(\frac{X_i}{n}\right) -r\left(\frac{X_i-1}{n}\right) \right)^2 \nonumber \\
 & \qquad +\sum_{1\leq i\leq S} p_i\left(1-\frac{X_i}{n}\right) \left(  r\left(\frac{X_i + 1}{n} \right)-r\left(\frac{X_i}{n}\right) \right)^2 .
\end{align}

Note that we interpret $\frac{X_i}{n} \left(  r\left(\frac{X_i}{n}\right) -r\left(\frac{X_i-1}{n}\right) \right)^2$ as $0$ when $X_i = 0$. Taking expectations of $\bE [D_-^2|X_1^S]$ and $\bE [D_+^2| X_1^S]$ with respect to $X_1^S$, we have
\begin{align}
\bE [D_-^2] & = \sum_{1\leq i\leq S} \sum_{1\leq j\leq n} \frac{j}{n} \left( r\left(\frac{j-1}{n} \right) - r\left(\frac{j}{n}\right)\right)^2 \nonumber \\
& \quad \quad \times  \bP(\mathsf{B}(n,p_i) = j)
\end{align}
and
\begin{align}
& \bE[D_+^2] \nonumber \\
 & \quad = \sum_{1\leq i\leq S} \sum_{0\leq j\leq n} \Bigg( \frac{j}{n} \left(r\left(\frac{j-1}{n} \right)- r\left(\frac{j}{n}\right)\right)^2 \nonumber \\
 & \quad \quad + \left(1-\frac{j}{n}\right) \left(r\left(\frac{j}{n}\right) - r\left(\frac{j+1}{n}\right)\right)^2 \Bigg) \nonumber \\
 & \quad \quad \times p_i \bP(\mathsf{B}(n,p_i) = j).
\end{align}

After some algebra, one can show that $\bE[D_+^2] = \bE[D_-^2]$. It then follows from (\ref{eqn.entropyesnew}) that
\begin{align}
& \mathsf{Var}(H(P_n)) \nonumber \\
& \quad \leq \frac{n}{2} \cdot \bE \left( \hat{H}(Z_1,Z_2,\ldots,Z_n) - \hat{H}(Z_1',Z_2,\ldots,Z_n) \right)^2 \\
& \quad= \frac{n}{2} \cdot \bE \left( D_- + D_+ \right)^2 \\
& \quad\leq n \cdot \bE (D_-^2 + D_+^2) \\
& \quad\leq 2n \bE D_-^2 \\
& \quad = 2n \cdot \sum_{1\leq i\leq S} \bE P_n(i) \left( r(P_n(i)) - r(P_n(i)-\frac{1}{n}) \right)^2 \label{eqn.espaninskifurther}
\end{align}

The proof above is an elaborate version of that in \cite[App. B.3]{Paninski2003}. Now we proceed to obtain non-asymptotic upper bounds of (\ref{eqn.espaninskifurther}). For $x\geq 1/n$, it follows from Taylor expansion with integral form residue that
\begin{align}
(x-\frac{1}{n}) \ln (x-\frac{1}{n}) & = x\ln x + (\ln x + 1) (-\frac{1}{n}) \nonumber \\
& \quad + \int_x^{x-\frac{1}{n}} (x-\frac{1}{n}-u) \frac{1}{u} du.
\end{align}
Then, we have
\begin{align}
& \left|r(P_n(i)) - r\left(P_n(i) -\frac{1}{n}\right)\right| \nonumber \\
& \quad\leq \frac{|\ln P_n(i) + 1|}{n} + \left| \int_{P_n(i)}^{P_n(i)-\frac{1}{n}} \frac{P_n(i)-\frac{1}{n}}{u}du \right| + \frac{1}{n} \\
& \quad \leq \frac{|\ln P_n(i) + 1|+2}{n}.
\end{align}

Hence, we have
\begin{equation}
\mathsf{Var}(H(P_n)) \leq 2n\cdot\sum_{1\leq i\leq S} \bE P_n(i) \left( \frac{|\ln P_n(i) +1|+2}{n}\right)^2.
\end{equation}
Noting that $\ln P_n(i) \leq 0$, hence $0\leq |\ln P_n(i) + 1| \leq 1 - \ln P_n(i)$. We have
\begin{align}
\mathsf{Var}(H(P_n)) & \leq \frac{2}{n}\cdot\sum_{1\leq i\leq S} \bE P_n(i) \left( \ln P_n(i)-3 \right)^2 \\
& \leq \frac{2}{n} \sum_{1\leq i\leq S} p_i(\ln p_i - 3)^2 \\
& \leq \frac{2}{n} S \cdot \frac{1}{S} (-\ln S -3)^2 \\
& = \frac{2(\ln S + 3)^2}{n},
\end{align}
where we have used the fact that $x(\ln x -3)^2$ is a concave function on $[0,1]$.

\subsection{Proof of Lemma~\ref{lemma.hbvariance}}
  We apply the bounded differences inequality (Lemma~\ref{lemma.es}). In our case, $F(\hat{P}_B)$ is a function of $n$ independent random variables $\{Z_i\}_{1\le i\le n}$ taking values in $\cZ=\{1,2,\cdots,S\}$. Changing one location of the sample would make some symbol with count $j$ to have count $j+1$, and another symbol with count $i$ to have count $i-1$. Then the absolute value of the total change in the functional estimator is
  \begin{align}
 &    \Bigg|f\left(\frac{j+1+a}{n+Sa}\right) - f\left(\frac{j+a}{n+Sa}\right) \nonumber \\
 & \quad - f\left(\frac{i+a}{n+Sa}\right) + f\left(\frac{i-1+a}{n+Sa}\right)\Bigg| \\
 & \quad   \le 2\max_{1\le k\le n} \left|f\left(\frac{k+a}{n+Sa}\right) - f\left(\frac{k-1+a}{n+Sa}\right)\right|.
  \end{align}

 In light of the Taylor expansion with integral form residue, we have that for $1\ge x\ge t>0$,
 \begin{align}
   (x-t)\ln(x-t) = x\ln x -t (\ln x+1) + \int_x^{x-t} \frac{x-t-u}{u}du
 \end{align}
 so
 \begin{align}
   |(x-t)\ln(x-t) - x\ln x| & \le t|\ln x+1| \nonumber \\
   & \quad + \left|\int_x^{x-t} \frac{x-t}{u}du\right| + t \\
   & \le t|\ln x+1| + 2t \\
   & \le t(3-\ln x).\label{eq:entropy_inequality}
 \end{align}
As a result,
  \begin{align}
    & \max_{1\le k\le n}\left|f\left(\frac{k+a}{n+Sa}\right) - f\left(\frac{k-1+a}{n+Sa}\right)\right| \nonumber \\
    &\quad\le \max_{1\le k\le n}\frac{1}{n+Sa}\left(3-\ln\left(\frac{k+a}{n+Sa}\right) \right) \\
    &\quad \le \frac{1}{n+Sa}\left(3+\ln \left(\frac{n+Sa}{a+1}\right)\right).
  \end{align}
  Hence, the bounded differences inequality shows that
  \begin{align}
   & \mathsf{Var}\left(H(\hat{P}_B)\right) \nonumber \\
   & \quad\le n\max_{2\le k\le n}\left(f\left(\frac{k+a}{n+Sa}\right) - f\left(\frac{k-1+a}{n+Sa}\right)\right)^2 \\
   & \quad \le \frac{n}{(n+Sa)^2}\left(3+\ln\left(\frac{n+Sa}{a+1}\right)\right)^2,
  \end{align}
  which completes the proof of the first part.

To prove the second part, we use the Efron--Stein inequality~(Lemma~\ref{lemma.esnew}). Since $H(\hat{P}_B)=\hat{H}_B(Z_1,\cdots,Z_n)$ is invariant to any permutation of $(Z_1,Z_2,\cdots,Z_n)$, we know that the Efron--Stein inequality implies
 \begin{align}
 &  \mathsf{Var}\left(H(\hat{P}_B)\right)  \nonumber \\
 & \quad \le \frac{n}{2}\bE\left(\hat{H}_B(Z_1',Z_2,\cdots,Z_n)-\hat{H}_B(Z_1,Z_2,\cdots,Z_n)\right)^2,
 \end{align}
 where $Z_1'$ is an i.i.d. copy of $Z_1$.

Recall that
 \begin{align}
   X_i = \sum_{j=1}^n \mathbbm{1}(Z_j=i),\qquad 1\le i\le S.
 \end{align}

 For brevity, we denote the $S$-tuple $(X_1,\cdots,X_S)$ as $X_1^S$, and the $n$-tuple $(Z_1,\cdots,Z_n)$ as $Z_1^n$. A specific realization of $(X_1,\cdots,X_S)$ is denoted by $x_1^S=(x_1,\cdots,x_S)$, and a specific realization of $(Z_1,\cdots,Z_n)$ is denoted by $z_1^n=(z_1,\cdots,z_n)$. Then we have
 \begin{align}
   &\bE\left(\hat{H}_B(Z_1',Z_2,\cdots,Z_n)-\hat{H}_B(Z_1,Z_2,\cdots,Z_n)\right)^2\\ &=
   \sum_{x_1^S} \bP(X_1^S=x_1^S) \nonumber \\
   & \quad \times \bE\left[\left(\hat{H}_B(Z_1',Z_2,\cdots,Z_n)-\hat{H}_B(Z_1,Z_2,\cdots,Z_n)\right)^2\left|X_1^S=x_1^S\right.\right].
 \end{align}

 In light of Lemma~\ref{lemma.conditionalpaninski}, we know that the conditional distribution of $Z_1$ conditioned on $(X_1,\cdots,X_S)$ is the discrete distribution $(X_1/n,X_2/n,\cdots,X_S/n)$. Denoting $r(p)=f(\frac{np+a}{n+Sa})$, we can rewrite
 \begin{align}
   \hat{H}_B(Z_1',Z_2,\cdots,Z_n)-\hat{H}_B(Z_1,Z_2,\cdots,Z_n) = D_- + D_+
 \end{align}
 where
 \begin{align}
   D_- &= r\left(\frac{X_{Z_1}-1}{n}\right) - r\left(\frac{X_{Z_1}}{n}\right)\\
   D_+ &= \begin{cases}
     r\left(\frac{X_{Z_1'}+1}{n}\right) - r\left(\frac{X_{Z_1'}}{n}\right)&Z_1\neq Z_1'\\
     r\left(\frac{X_{Z_1'}}{n}\right) - r\left(\frac{X_{Z_1'}-1}{n}\right)&Z_1= Z_1'
   \end{cases}.
 \end{align}

 Here, $D_-$ is the change in $\hat{H}_B$ that occurs when $Z_1$ is removed according to the distribution $(X_1/n,X_2/n,\cdots,X_S/n)$, and $D_+$ is the change in $\hat{H}$ that occurs when $Z_1'$ is added back according to the true distribution $P$. Now we have
 \begin{align}
   \bE[D_-^2|X_1^S] &= \sum_{i=1}^S \frac{X_i}{n}\left(r\left(\frac{X_{i}-1}{n}\right)-r\left(\frac{X_{i}}{n}\right)\right)^2\\
   \bE[D_+^2|X_1^S] &= \sum_{i=1}^S p_i\frac{X_i}{n}\left(r\left(\frac{X_{i}}{n}\right)-r\left(\frac{X_{i}-1}{n}\right)\right)^2 \nonumber \\
   & \quad   + \sum_{i=1}^S p_i\left(1-\frac{X_i}{n}\right)\left(r\left(\frac{X_{i}+1}{n}\right)-r\left(\frac{X_{i}}{n}\right)\right)^2
 \end{align}
 where we define $r(x)=0$ when $x\notin[0,1]$. Then, by the law of iterated expectation, we know that
 \begin{align}
   \bE[D_-^2] &= \sum_{i=1}^S \sum_{j=0}^n\frac{j}{n}\left(r\left(\frac{j-1}{n}\right)-r\left(\frac{j}{n}\right)\right)^2 \nonumber \\
   & \quad \times\bP(\mathsf{B}(n,p_i)=j)\\
   \bE[D_+^2] &= \sum_{i=1}^S \sum_{j=0}^n\Bigg(\frac{j}{n}\left(r\left(\frac{j-1}{n}\right)-r\left(\frac{j}{n}\right)\right)^2  \nonumber \\
   & \quad + \left(1-\frac{j}{n}\right) \left(r\left(\frac{j+1}{n}\right)-r\left(\frac{j}{n}\right)\right)^2 \Bigg) \nonumber \\
   & \quad \times p_i\bP(\mathsf{B}(n,p_i)=j). 
 \end{align}

 After some algebra we can show that $\bE[D_-^2]=\bE[D_+^2]$. Hence, we have
 \begin{align}
 &  \mathsf{Var}\left(H(\hat{P}_B)\right) \nonumber \\
  & \quad\le \frac{n}{2}\bE\left(D_-+D_+\right)^2\le n\bE\left(D_-^2+D_+^2\right) = 2n\bE D_-^2 \\
   &\quad= 2n\sum_{i=1}^S\bE P_n(i)\left(r(P_n(i)-\frac{1}{n})-r(P_n(i))\right)^2\\
   &\quad\le 2n\sum_{i=1}^S \bE P_n(i)\left(\frac{1}{n+Sa}\left[3-\ln \left(\frac{nP_n(i)+a}{n+Sa}\right)\right]\right)^2\\
   &\quad= \frac{2n}{(n+Sa)^2}\sum_{i=1}^S \bE P_n(i)\left(3-\ln \left(\frac{nP_n(i)+a}{n+Sa}\right)\right)^2\\
   &\quad\le \frac{2n}{(n+Sa)^2}\sum_{i=1}^S p_i\left(3-\ln \left(\frac{np_i+a}{n+Sa}\right)\right)^2\\
   &\quad\le \frac{2n}{(n+Sa)^2}S\cdot \frac{1}{S}\left(3-\ln \left(\frac{n/S+a}{n+Sa}\right)\right)^2\\
   &\quad= \frac{2n(3+\ln S)^2}{(n+Sa)^2}
 \end{align}
 where we have used the inequality (\ref{eq:entropy_inequality}) and Jensen's inequality due to
 \begin{align}
 &  \frac{d^2}{dx^2}\left[x\left(\ln \left(\frac{nx+a}{n+Sa}\right)-3\right)^2\right] \nonumber \\
 & \quad = \frac{n}{nx+a}\Bigg[3\left(\ln \left(\frac{nx+a}{n+Sa}\right)-3\right)\nonumber \\
 & \qquad +\frac{nx}{nx+a}\left(4-\ln \left(\frac{nx+a}{n+Sa}\right)\right)\Bigg] \\
 & \quad <0.
 \end{align}

\subsection{Proof of Lemma~\ref{lemma.biasalphalarge}}

It is well known (see, e.g. \cite[Cor. 10.4.2]{Devore--Lorentz1993}) that if $f$ is concave in $(0,1)$, then
\begin{equation}
f(x) - B_n[f](x) \geq 0,\quad 0\leq x\leq 1.
\end{equation}

Hence we focus on deriving the other bound. For concave function $f_\alpha(x) = -x^\alpha,\alpha\in (1,2) $, Taylor's polynomial of degree $5$ at $x = x_0$ takes the form
\begin{align*}
Q_5(x) & = -\frac{\alpha(\alpha-1)}{2} x_0^{\alpha-2} (x-x_0)^2 \nonumber \\
& \quad - \frac{\alpha(\alpha-1)(\alpha-2)}{6} x_0^{\alpha-3} (x-x_0)^3 \nonumber  \\
& \quad -\frac{\alpha(\alpha-1)(\alpha-2)(\alpha-3)}{24} x_0^{\alpha-4}(x-x_0)^4 \nonumber \\
& \quad - \frac{\alpha(\alpha-1)(\alpha-2)(\alpha-3)(\alpha-4)}{120} x_0^{\alpha-5} (x-x_0)^5 \nonumber \\
& \quad + \textrm{affine terms of }x
\end{align*}

We know that the Bernstein polynomial of any affine function on $[0,1]$ is the affine function itself, hence it suffices to consider the non-affine part of $Q_5(x)$. \cite[Prop. 4]{Braess--Sauer2004} showed the following results for Bernstein polynomials:
\begin{lemma}\label{lemma.brastaylor}
Let $0\leq x_0 \leq 1$. Then we have
\begin{align}
B_n[(x-x_0)^2](x_0) & =  \frac{x_0(1-x_0)}{n} \\
B_n[(x-x_0)^3](x_0) & = \frac{x_0(1-x_0)}{n^2}(1-2x_0) \\
B_n[(x-x_0)^4](x_0) & = 3 \frac{x_0^2(1-x_0)^2}{n^2} \nonumber \\
& \quad + \frac{x_0(1-x_0)}{n^3} [ 1-6x_0(1-x_0)] \\
B_n[(x-x_0)^5](x_0) & = \Bigg( 10 \frac{x_0^2(1-x_0)^2}{n^3} \nonumber \\
& \quad + \frac{x_0(1-x_0)}{n^4}[1-12 x_0(1-x_0)] \Bigg) \nonumber \\
& \qquad \times (1-2x_0)
\end{align}
\end{lemma}

Applying Lemma~\ref{lemma.qn1lower} and Lemma~\ref{lemma.brastaylor}, taking $x_0 =x$, we have the desired bound. 

\subsection{Proof of Lemma~\ref{lemma.alphalower}}

It is well known (see, e.g. \cite[Cor. 10.4.2]{Devore--Lorentz1993}) that if $f$ is concave in $(0,1)$, then
\begin{equation}
f(x) - B_n[f](x) \geq 0,\quad 0\leq x\leq 1.
\end{equation}

Hence we focus on deriving the other bound. For function $f_\alpha(x) = x^\alpha$, Taylor's polynomial of degree $3$ at $x = x_0$ takes the form
\begin{align}
Q_3(x) & = \frac{\alpha(\alpha-1)}{2} x_0^{\alpha-2} (x-x_0)^2 \nonumber \\
& \quad  + \frac{\alpha(\alpha-1)(\alpha-2)}{6} x_0^{\alpha-3} (x-x_0)^3 \nonumber \\
& \quad + \textrm{ affine terms of }x.
\end{align}

Applying Lemma~\ref{lemma.qn1lower} and Lemma~\ref{lemma.brastaylor}, taking $x_0 =x$, we have
\begin{align}
& Q_3(x) - B_n[Q_3](x) \nonumber \\
& \quad = - \frac{\alpha(\alpha-1)}{2}x^{\alpha-2} \frac{x(1-x)}{n} \nonumber \\
& \quad \quad - \frac{\alpha(\alpha-1)(\alpha-2)}{6}x^{\alpha-3} \frac{x(1-x)}{n^2}(1-2x) \\
& \quad= \frac{\alpha(1-\alpha)}{2n} x^{\alpha-2}(1-x) \left( x - \frac{2-\alpha}{3n} (1-2x) \right) \\
& \quad \geq \frac{\alpha(1-\alpha)}{2n}x^{\alpha-2} (1-x) \left( x - \frac{2-\alpha}{3n} \right).
\end{align}

Hence, we have
\begin{align}
f_\alpha(x) - B_n[f_\alpha](x) &  \geq Q_3(x) - B_n[Q_3](x) \\
&  \geq \frac{\alpha(1-\alpha)}{2n}x^{\alpha-2} (1-x) \left( x - \frac{2-\alpha}{3n} \right).
\end{align}

\subsection{Proof of Lemma~\ref{lemma.hbbiaslowerbound}}

  By setting $P=(1,0,0,\cdots,0)$, we have $H(P)=0$ and
  \begin{align}
    H(\hat{P}_B) & = -\frac{(S-1)a}{n+Sa}\ln\left(\frac{a}{n+Sa}\right) -\frac{n+a}{n+Sa}\ln\left(\frac{n+a}{n+Sa}\right) \\
    & \ge \frac{(S-1)a}{n+Sa}\ln\left(\frac{n+Sa}{a}\right),
  \end{align}
  hence we have obtained the first lower bound
  \begin{align}
  \sup_{P\in\mathcal{M}_S}\left|\bE_P H(\hat{P}_B) - H(P)\right| \ge \frac{(S-1)a}{n+Sa}\ln\left(\frac{n+Sa}{a}\right).
  \end{align}
  
If $n<Sa$, then
\begin{align}
  \sup_{P\in\mathcal{M}_S}\left|\bE_P H(\hat{P}_B) - H(P)\right|& \ge \frac{(S-1)a}{2Sa}\ln S \\
  & \geq  \frac{S-1}{2S} \ln S. 
\end{align}

If $n > 2ea$, then
\begin{align}
 \sup_{P\in\mathcal{M}_S}\left|\bE_P H(\hat{P}_B) - H(P)\right| & \ge \frac{(S-1)a}{Sa + 2ea}\ln S \\
 & = \frac{S-1}{S+2e} \ln S.  
\end{align}

From now on we assume $n\geq Sa, n\geq 2ea$. For $n\ge 15S$, it follows from applying Lemma~\ref{lemma.glower} that
  \begin{align} \label{eqn.mlepaperlowerbound}
    \sup_{P\in\mathcal{M}_S}\left|\bE_P H(\hat{P}) - H(P)\right| \ge \frac{S-1}{2n} + \frac{S^2}{20n^2} - \frac{1}{12n^2}.
  \end{align}
If $n<15S$, then it follows from applying Lemma~\ref{lemma.glower} that one can essentially take $S = \lfloor n/15 \rfloor$ in (\ref{eqn.mlepaperlowerbound}), and obtain
\begin{equation}
    \sup_{P\in\mathcal{M}_S}\left|\bE_P H(\hat{P}) - H(P)\right| \ge \frac{\lfloor n/15 \rfloor}{2n} - \frac{1}{4n}.
\end{equation}  

It follows from a refinement result of Cover and Thomas~\cite[Thm. 17.3.3]{Cover--Thomas2006} that when $| \hat{p}_{B,i} - \hat{p}_i | \leq 1/2$ for all $i$ (which is ensured by condition $n\geq Sa$), we have
\begin{align}
 |H(\hat{P}_B) - H(\hat{P})| \leq S f \left( \frac{\| \hat{P}_B - \hat{P} \|_1}{S} \right),
\end{align}
where $f(x) = -x\ln x, x \in [0,1]$. We have
\begin{align}
\frac{1}{S} \| \hat{P}_B - \hat{P} \|_1 & = \frac{1}{S} \sum_{i = 1}^S \frac{|S \hat{p}_i - 1|a}{n+Sa} \\
& \leq \frac{2(S-1)a}{S(n+Sa)}, 
\end{align}
where the last step follows from~(\ref{eq:convexcornerpoints}). Since we have assumed $n\geq 2ea$, we have $\frac{2a}{n+Sa} \leq \frac{2a}{n}\leq e^{-1}$. Since the function $f(x) = -x\ln x,x\in [0,1]$ is monotonically increasing when $x\in [0,e^{-1}]$, we have
\begin{align}
 |H(\hat{P}_B) - H(\hat{P})| \le \frac{2(S-1)a}{n+Sa}\ln\left(\frac{n+Sa}{a}\right).
\end{align}

A combination of these two inequalities yield the second lower bound
\begin{align}
&  \sup_{P\in\mathcal{M}_S}\left|\bE_P H(\hat{P}_B) - H(P)\right | \nonumber \\
& \quad \ge  \frac{S-1}{2n} + \frac{S^2}{20n^2} - \frac{1}{12n^2} - \frac{2(S-1)a}{n+Sa}\ln\left(\frac{n+Sa}{a}\right)
\end{align}
when $n\geq 15S$, and the second lower bound
\begin{align}
&  \sup_{P\in\mathcal{M}_S}\left|\bE_P H(\hat{P}_B) - H(P)\right | \nonumber  \\
& \quad \ge  \frac{\lfloor n/15 \rfloor}{2n} - \frac{1}{4n} - \frac{2(S-1)a}{n+Sa}\ln\left(\frac{n+Sa}{a}\right)
\end{align}
when $n<15S$. 

  Hence we are done by using these two lower bounds and the inequality $\max\{a,b\}\ge \frac{3a+b}{4}$.
\section{Proofs of auxiliary lemmas}

\subsection{Proof of Lemma \ref{lem_variance_p_small}}
Since $nX$ is an integer, we have $(nX)^2\ge (nX)^{2\alpha},0<\alpha<1$. Hence, for $p\leq 1/n$,
\begin{align}
  \mathsf{Var}(X^\alpha) & \le \bE X^{2\alpha} \\
  & = \frac{\bE (nX)^{2\alpha}}{n^{2\alpha}} \\
  &  \le  \frac{\bE (nX)^{2}}{n^{2\alpha}} \\
  & = \frac{(np)^2+np(1-p)}{n^{2\alpha}} \\
  &  \le \frac{2}{n^{2\alpha}}\wedge \frac{2p}{n^{2\alpha-1}},
\end{align}
which completes the proof.

\subsection{Proof of Lemma \ref{lem_variance_p_big}}
Denoting $f(p)=p^\alpha,0<\alpha<1$, we have
\begin{align}
& \mathsf{Var}(f(X)) \nonumber \\
& \quad = \bE f^2(X) - (\bE f(X))^2 \\
& \quad= \bE f^2(X) - f^2(p) + f^2(p) - (\bE f(X))^2 \\
& \quad \leq | \bE f^2(X) - f^2(p)|+ | f^2(p) - (\bE f(X)-f(p) + f(p))^2 |\\
& \quad = | \bE f^2(X) - f^2(p)|+ |(\bE f(X)-f(p))^2 \nonumber \\
& \qquad+ 2f(p) (\bE f(X)-f(p))| \\
& \quad \leq | \bE f^2(X) - f^2(p)| + |\bE f(X)-f(p)|^2 \nonumber \\
& \qquad + 2f(p)|\bE f(X)-f(p)|. \label{eqn.varianceestimatelarge}
\end{align}

Hence, it suffices to obtain bounds on $| \bE f^2(X) - f^2(p)|$ and $|\bE f(X) - f(p)|$. Denoting $r(x) = f^2(x)$, it follows from Taylor's formula and the integral representation of the remainder term that
\begin{align}
r(X) & = f^2(p) + r'(p)(X-p) + R_1(X;p) \\
R_1(X;p)& =\int_p^X (X-u)r''(u)du = \frac{1}{2}r''(\eta_X)(X-p)^2
\end{align}
where $\eta_X \in [\min\{X,p\},\max\{X,p\}]$. 

Similarly, we have
\begin{align}
f(X) & = f(p) + f'(p)(X-p) + R_2(X;p) \\
R_2(X;p) & = \int_p^X (X-u)f''(u)du = \frac{1}{2} f''(\nu_X)(X-p)^2,
\end{align}
where $\nu_X\in [\min\{X,p\},\max\{X,p\}]$. 

Taking expectation on both sides with respect to $X$, where $nX \sim \mathsf{B}(n,p)$, we have
\begin{equation}\label{eqn.u2biaslarge}
|\bE f^2(X) - f^2(p)| = |\bE R_1(X;p) |.
\end{equation}
Similarly, we have
\begin{equation}
|\bE f(X) - f(p)| =  | \bE R_2(X;p)|.
\end{equation}

It is straightforward to show that
\begin{align}
  |r''(x)|&=2\alpha(2\alpha-1)x^{2\alpha-2} \le 2x^{2\alpha-2},\\
  |f''(x)|&=\alpha(1-\alpha)x^{\alpha-2} \le \frac{1}{4}x^{\alpha-2}.
\end{align}

Now we are in the position to bound $|\bE R_1(X;p)|$ and $|\bE R_2(X;p)|$. For $|\bE R_1(X;p)|$, we have
  \begin{align}
& |\bE R_1(X;p)| \nonumber \\
& \quad\leq \bE |R_1(X;p)| \\
& \quad= \bE [|R_1(X;p)\mathbbm{1}(X\geq p/2)|] + \bE[R_1(X;p)\mathbbm{1}(X<p/2)] \\
& \quad\leq \bE \left[ 2(p/2)^{2\alpha-2} (X-p)^2\right] + \bE[R_1(X;p)\mathbbm{1}(X<p/2)] \\
& \quad\leq 8\frac{p^{2\alpha-1}}{n} + \sup_{x\leq p/2}|R_1(x;p)|\bP(nX<np/2) \\
& \quad \leq 8\frac{p^{2\alpha-1}}{n} + \sup_{x\leq p/2}|R_1(x;p)| e^{-np/8},
\end{align}
where in the last step we have used Lemma \ref{lem_chernoff}. Regarding $\sup_{x\leq p/2}|R_1(x;p)|$, for any $x\leq p/2$, we have
\begin{align}
R_1(x;p) & = \int_x^p (u-x)r''(u)du \leq \int_x^p (u-x) 2u^{2\alpha-2}du \\
 &\leq 2 \int_x^p u^{2\alpha-1}du \\
 &\leq 2 \int_0^p u^{2\alpha-1}du\\
 &= \frac{p^{2\alpha}}{\alpha}.
\end{align}

Hence, we have
\begin{equation}
|\bE R_1(X;p)| \leq \frac{8p^{2\alpha-1}}{n} + \frac{1}{\alpha}p^{2\alpha}e^{-np/8}.
\end{equation}

Analogously, we obtain the following bound for $|\bE R_2(X;p)|$:
\begin{equation}
|\bE R_2(X;p)|\leq \frac{p^{\alpha-1}}{n} + \frac{1}{4\alpha}p^\alpha e^{-np/8}.
\end{equation}

Plugging these estimates of $|\bE R_1(X;p)|$ and $|\bE R_2(X;p)|$ into (\ref{eqn.varianceestimatelarge}), we have for $p\geq 1/n$,
\begin{align}
\mathsf{Var}(X^\alpha)  &\leq \frac{8p^{2\alpha-1}}{n} + \frac{1}{\alpha}p^{2\alpha}e^{-np/8} \nonumber \\
& \quad +  \left(\frac{p^{\alpha-1}}{n} + \frac{1}{4\alpha} p^\alpha e^{-np/8} \right)^2 \nonumber \\
& \quad +2f(p)\left(\frac{p^{\alpha-1}}{n} + \frac{1}{4\alpha}p^\alpha e^{-np/8} \right)\\
&\le \frac{8p^{2\alpha-1}}{n} + \frac{1}{\alpha}p^{2\alpha}e^{-np/8} + \frac{2p^{2(\alpha-1)}}{n^2} \nonumber \\
& \quad  + \frac{1}{8\alpha^2}p^{2\alpha} e^{-np/4}+2p^\alpha\left(\frac{p^{\alpha-1}}{n} + \frac{1}{4\alpha}p^\alpha e^{-np/8} \right)\\
&\le \frac{10p^{2\alpha-1}}{n} + \frac{3}{2\alpha}p^{2\alpha}e^{-np/8} + \frac{2}{n^{2\alpha}} + \frac{1}{8\alpha^2}p^{2\alpha} e^{-np/4}\\
&\le \frac{10p^{2\alpha-1}}{n} + \frac{3}{2\alpha}\left(\frac{16\alpha}{en}\right)^{2\alpha} + \frac{2}{n^{2\alpha}} + \frac{1}{8\alpha^2}\left(\frac{8\alpha}{en}\right)^{2\alpha}.
\end{align}
where we have used the following inequality in the last step: for $x\in(0,1)$ and any $c>0$,
\begin{align}
  x^{2\alpha}e^{-cnx} \le \left(\frac{2\alpha}{cen}\right)^{2\alpha}.
\end{align}

Note that if $0<\alpha<1/2$, we can upper bound $\frac{10p^{2\alpha-1}}{n}$ by $\frac{10}{n^{2\alpha}}$, since we have constrained $p\geq \frac{1}{n}$.

\bibliographystyle{IEEEtran}
\bibliography{di}

\begin{IEEEbiographynophoto}{Jiantao Jiao}
(S'13) received the B.Eng. degree with the highest honor in Electronic Engineering from Tsinghua University, Beijing, China in 2012, and a Master's degree in Electrical Engineering from Stanford University in 2014. He is currently working towards the Ph.D. degree in the Department of Electrical Engineering at Stanford University. He is a recipient of the Stanford Graduate Fellowship (SGF). His research interests include information theory and statistical signal processing, with applications in communication, control,
computation, networking, data compression, and learning. 
\end{IEEEbiographynophoto}

\begin{IEEEbiographynophoto}{Kartik Venkat}
(S'12) is currently a Research Associate in New York at PDT Partners, a quantitative investment manager.  Kartik's research interests include statistical inference, information theory, machine learning and their inter-connections. Kartik received his Ph.D. in Electrical Engineering from Stanford University in 2015. Kartik also received a Masterâ€™s Degree from the same university in 2012, and a Bachelorâ€™s degree from the Indian Institute of Technology Kanpur in 2010, both in Electrical Engineering. Kartik was the recipient of the Thomas M. Cover Dissertation Award in 2016 awarded by the IEEE. Kartik was named the 2015 Marconi Society Paul Baran Young Scholar. His other honors include the Jack Keil Wolf Student Best Paper Award at the 2012 International Symposium on Information Theory, a Stanford Graduate Fellowship for Engineering and Sciences, the Numerical Technologies Founders Graduate Prize, and the National Talent Scholarship awarded by the Government of India. 
\end{IEEEbiographynophoto}

\begin{IEEEbiographynophoto}{Yanjun Han}
(S'14) received his B.Eng. degree with the highest honor in
Electronic Engineering from Tsinghua University, Beijing, China in 2015, and a Master's degree in Electrical Engineering from Stanford University in 2017. He is currently working towards the Ph.D. degree in the Department of Electrical Engineering at Stanford University. His research interests include information theory and statistics, with applications in communications, data compression, and learning.
\end{IEEEbiographynophoto}

\begin{IEEEbiographynophoto}{Tsachy Weissman}
(S'99-M'02-SM'07-F'13) graduated summa cum laude with a
B.Sc. in electrical engineering from the Technion in 1997, and earned
his Ph.D. at the same place in 2001. He then worked at Hewlett-Packard
Laboratories with the information theory group until 2003, when he joined
Stanford University, where he is Professor of Electrical
Engineering and incumbent of the
STMicroelectronics chair in the School of Engineering.
He has spent leaves at the Technion, and at ETH Zurich.

Tsachy's research is focused on information theory, statistical signal
processing, the interplay between them, and their applications.

He is recipient of several best paper awards, and prizes for excellence in research.

He served on the editorial board of the \textsc{IEEE Transactions on Information Theory} from Sept. 2010 to Aug. 2013, and currently serves on the editorial board of Foundations and Trends in Communications and Information Theory.
\end{IEEEbiographynophoto}

\end{document}